\documentclass[12pt, letterpaper, onecolumn, draftcls]{IEEEtran}
\usepackage{indentfirst}
\usepackage{amsmath}
\usepackage{amssymb}
\usepackage{amsfonts}
\usepackage{dsfont}
\usepackage{cite}
\usepackage{times}
\usepackage[dvips]{graphicx}
\usepackage{comment}
\usepackage{epsf}
\usepackage{psfrag}
\usepackage{theorem}
\usepackage{geometry}
\usepackage[nolists, nomarkers]{endfloat}

\newcommand{\spa}{\,\,\,\!\!}

\newcommand{\figw}{0.44\columnwidth}

\newcommand{\figwd}{0.6\columnwidth}
\newcommand{\figwe}{0.34\columnwidth}

\newcommand{\vup}{\vspace{-1mm}}

\newcommand{\be}{\begin{equation}}
\newcommand{\ee}{\end{equation}}
\newcommand{\bea}{\begin{eqnarray}}
\newcommand{\eea}{\end{eqnarray}}

 \theoremstyle{break}   
 
 \newtheorem{teo}{Theorem}

 {\theorembodyfont{\normalfont}     
 }

\begin{document}
\sloppy
\setcounter{page}{0}
\baselineskip 21pt
\title{An Algorithmic Solution for\\ Computing Circle Intersection Areas and\\ its Applications to Wireless Communications}
\author{\vspace{3mm}Federico Librino, Marco Levorato and Michele Zorzi\\
Department of Information Engineering, University of Padova\\
E--mail: \texttt{\{librinof,levorato,zorzi\}@dei.unipd.it \vspace{-2mm} }
\thanks{This work was partially supported by Qualcomm, Incorporated. Part of this work has been presented at WiOpt 2009, Seoul, South Korea.} \\
}

\date{}
\maketitle
\pagestyle{empty}
\thispagestyle{empty}
\begin{abstract}
A novel iterative algorithm for the efficient computation of the intersection areas of an arbitrary number of circles
is presented. The algorithm, based on a trellis-structure, hinges on two geometric results which allow the existence-check
and the computation of the area of the intersection regions generated by more than three circles by simple algebraic manipulations
of the intersection areas of a smaller number of circles. The presented algorithm is a powerful tool for the performance
analysis of wireless networks, and finds many applications, ranging from sensor to cellular networks.
As an example of practical application, an insightful study of the uplink outage probability of in a wireless network with 
cooperative access points as a function of the transmission power and access point density is presented.
\end{abstract}

\section{Introduction and Problem Statement}
The computation of the intersection area of many circles is a challenging problem. While
the intersection of two circles is straightforward, even three circles admit several configurations,
each resulting in a different expression for the intersection area. Given the centers and the radii
of the circles, the automatic discrimination among the various cases requires involved condition testing.
If we consider cases with several circles the problem may appear unsolvable, as the close-form expressions for the
intersection areas become more and more involved and depend on the specific configuration among a huge number
of possibilities.

Despite the wide range of applications of this geometric problem, a systematic approach is still lacking. It has been addressed in the literature for three circles, but only for some specific configurations, in~\cite{trecerchi}. The intersections among $d$-dimensional balls are used to find their union in~\cite{uniball}, in a way analogous to the last step of our proposed algorithm, whereas \cite{equdisk} analyzes the special case of equal circular disks, showing that their intersection can be derived from intersection areas among fewer circles. However, no algorithmic solution is proposed to exploit the result in an organized and exhaustive way.

In this paper, an algorithm that efficiently computes the intersections of an arbitrary number of
circles is presented. The algorithm works in an iterative fashion and is based on a trellis structure. At each iteration, the existence of any
intersection is checked based on the areas computed in the previous steps, thus highly reducing the computational load.
Moreover, only the first three steps involve geometric considerations, whereas, when the number
of circles is higher than three, all the areas can be found via simple algebraic calculations.
The presented algorithm allows to efficiently solve configurations with many tens of circles, without any
assumption on the centers and radii of the considered circles.

The technical contributions of this paper are as follows:
\begin{itemize}
\item we derive two theorems, that provide an easy way to check the existence and calculate the area of the intersection
region of more than three circles, once the existence and the area of the intersections involving a smaller number of circles are known;
\item we present a trellis-based iterative algorithm that allows an easy computation of the wanted areas even for configurations with a large number of circles.
\end{itemize}

In the following, the applications of the presented tool in wireless networks' performance analysis are discussed, and then
the geometric problem is formally defined.  

\subsection{Applications to Wireless Networks}
Many frameworks for the evaluation, analysis and simulation of wireless communications are based on 
a signal propagation model in which the attenuation incurred by a transmitted signal is a monotonically decreasing
function of the distance from its source. Thus, the performance of a receiver is a function of its distance from
the source, and this leads to a characterization of wireless networks based on the concept of \emph{coverage range}.

The coverage range of a transmission can be defined by assigning a threshold bit error rate (BER), packet error rate (PER) or signal-to-noise-ratio (SNR) which determines an admissible region of received power. The coverage range then is the maximum distance
between two nodes which guarantees the received power to lie within the admissible region. The resulting coverage area
of a transmitter (receiver), is a circle centered on the receiver (transmitter) and with radius equal to the coverage range.
Note that different transmission power levels, packet encoding rate and, in general, transmission parameters
can be represented as multiple circles centered on the same node.

A node placed in a point of the plane covered by multiple coverage areas can communicate with all the nodes associated
with those coverage areas. The computation of the area of those regions enables a wide range of considerations in many
scenarios of interest.

In cellular networks, circular coverage areas may be used to design base station positioning in order to guarantee
connectivity~\cite{cellwirtz}. Recently, considerable attention has been devoted to the study of relaying strategies in
multihop cellular networks to improve capacity, coverage range and Quality-of-Service fairness (\emph{e.g.}, see~\cite{relay1,relay2,relay3}).
Circular coverage areas of base stations and relays \footnote{Multiple circles
associated with each base station/relay can be used to account for coverage shrinking as the number of mobiles increases.}
can be used to build a simple connectivity model aimed at the calculation of the overall capacity of the cell~\cite{relay5}.
The ability to calculate the area of the various intersections of the coverage areas granted by the proposed algorithm  
may be used to compute the probability that a mobile falls within coverage of a certain set of base stations/relays.

The computation of the intersection areas may also be used to model connectivity in many other infrastructured network scenarios. For instance, in
heterogeneous networks, the areas covered by different network infrastructures (GSM, UMTS, local area networks,
and so on) may intersect. Thus, the areas covered by multiple technologies may be used in order to allocate users and compute the average performance. This problem has been recently investigated in~\cite{stocgeo} for downlink \textit{K}--tier cellular networks.

In non-infrastructured ad hoc networks, circles have been traditionally used to characterize channel sensing and data packet decoding.
Again, given a topology, the area of the regions in which a new transmitter detects/decodes signals from the various sources can be computed
using the proposed algorithm.\footnote{For instance, the algorithm can be directly applied to computing the probability that a node with uniform spatial distribution falls within a region connecting other nodes, or becomes a hidden or exposed terminal.}

In sensor networks, localization relies on the reception
of beacons sent by nodes whose positions are known. The accuracy achieved by the
localization algorithm depends on the number of beacon sources that the node can hear. This requires
the computation of the probability that a node falls within an area covered by a certain number
of circles. Furthermore, intersections of multiple circles are also found when addressing the problem of
preserving complete sensing coverage of a certain area and connectivity~\cite{scov,cerchicover}.  

Another important example in which intersection areas are a fundamental aspect of the performance analysis is routing~\cite{svs1}. The
intersection of the circles may represent the area in which a user can provide connectivity to
some nodes of the network (corresponding to the centers of the various circles). When considering geographic packet
forwarding~\cite{paoloroute,direrout}, intersection areas may be helpful to derive the distribution of the advancement and the success probability
of the communication.

In this paper, as an example of application of the presented tool, we study the uplink outage probability in a wireless network with 
cooperative access points as a function of the transmission power and access point density. An analogous scenario can be found in cellular networks, where 
recent work showed that cooperation among Base Stations may offer considerable performance gain.
Multi-cell processing (MCP) has been proven to grant higher throughput and achievable data rate~\cite{simeone1,simeone2}, depending on the
topology as well as on the robustness of the backhaul links. Capacity may be also increased, as was shown in~\cite{capBScoop}, when
cooperation is aimed at cancelling interference. The use of relays, together with cooperating Base Stations, has been also considered
in~\cite{simeone3}. It is clear that the performance of MCP depends on how many access points are able to receive and
decode the transmission from a given source, which can be statistically determined by finding the intersections of the
coverage areas of Base Stations and relays.

In the following section, we state in detail the addressed geometric problem and the contribution of the paper.

\subsection{Problem Statement and Contribution}
\label{probstat}
Consider a set $\mathcal{C}=\{\gamma_1,\gamma_2,\ldots,\gamma_{N_c}\}$ of $N_c$ circles, whose centers and radii are known. 
The circles in $\mathcal{C}$ may partially overlap. We denote with 
\begin{equation}
\!\!\boldsymbol{\mathcal{I}}^{(n)}{=}\{\mathcal{I}^{(n)}_{\{i_1,\ldots,i_n\}},i_1,\ldots,i_n{\in}\{1,\ldots,N_c\}, i_j{\neq}i_u, {\rm for} j{\neq}u\}
\end{equation}
the set of all the possible intersection regions generated by $n$ circles, where $\mathcal{I}^{(n)}_{\{i_1,\ldots,i_n\}}{=}\bigcap_{i{\in}\{i_1,\ldots,i_n\}}\gamma_{i}$ is the set of the points that belong to all circles in $\{\gamma_{i_1},\gamma_{i_2},\ldots,\gamma_{i_n}\}{\subseteq}\mathcal{C}$. The set $\boldsymbol{\mathcal{I}}^{(1)}$ contains the circles in $\mathcal{C}$. We also define the notation $\mathcal{I}^{(n)}(i_1, \ldots, i_{N_c-n})$ to denote the intersection of $n$ circles out of the $N_c$ in $\mathcal{C}$ where circles $i_1, \ldots, i_{N_c-n}$ are not considered, i.e., $\mathcal{I}^{(n)}(i_1, \ldots, i_{N_c-n}){=}\bigcap_{i{\in}\{1,\ldots,N_c\}\setminus\{i_1,\ldots,i_{N_c-n}\}}\gamma_{i}$, which is hence equivalent to $\mathcal{I}^{(n)}_{\{1,\ldots,N_c\}\setminus\{i_1,\ldots,i_{N_c-n}\}}$.

However, these intersections are not disjoint regions of the plane. See for instance Fig.~\ref{fig:c2}, where a configuration with three
circles is depicted. In the figure, $\mathcal{I}^{(3)}_{\{1,2,3\}}{=}\mathcal{A}_1$,  $\mathcal{I}^{(2)}_{\{1,2\}}{=}\mathcal{A}_1{\cup}\mathcal{A}_4$,
$\mathcal{I}^{(2)}_{\{1,3\}}{=}\mathcal{A}_1{\cup}\mathcal{A}_2$ and $\mathcal{I}^{(2)}_{\{2,3\}}{=}\mathcal{A}_1{\cup}\mathcal{A}_3$.

We call the regions $\mathcal{A}_i$ in the figure \emph{exclusive intersection regions}, as they correspond to the intersection of a
certain subset of circles, excluding the regions covered by the other circles in $\mathcal{C}$. We denote these regions as $\mathcal{E}^{(n)}(i_1,\ldots,i_{N_c-n})$, where
\begin{equation}
\mathcal{E}^{(n)}(i_1,\ldots,i_{N_c{-}n})=\mathcal{I}^{(n)}(i_1,\ldots,i_{N_c{-}n}){\setminus}\bigcup_{e\in\{i_1,\ldots,i_{N_c{-}n}\}}\gamma_{e}.
\end{equation}  
For instance, $\mathcal{E}^{(1)}(2,3)$ is the region of the plane covered by $\gamma_1$ and that does not overlap with any other circle of $\mathcal{C}$
($\mathcal{A}_5$ in Fig.~\ref{fig:c2}), and $\mathcal{E}^{(2)}(3){=}\mathcal{A}_4$ is the intersection of $\gamma_1$ and $\gamma_2$,
excluding the area covered by $\gamma_3$. We define the set $\boldsymbol{\mathcal{E}}^{(n)}$
as the set of all the exclusive intersection regions generated by $n$ circles. 
\begin{figure}[t]
    \centering
    \includegraphics[width=\figw]{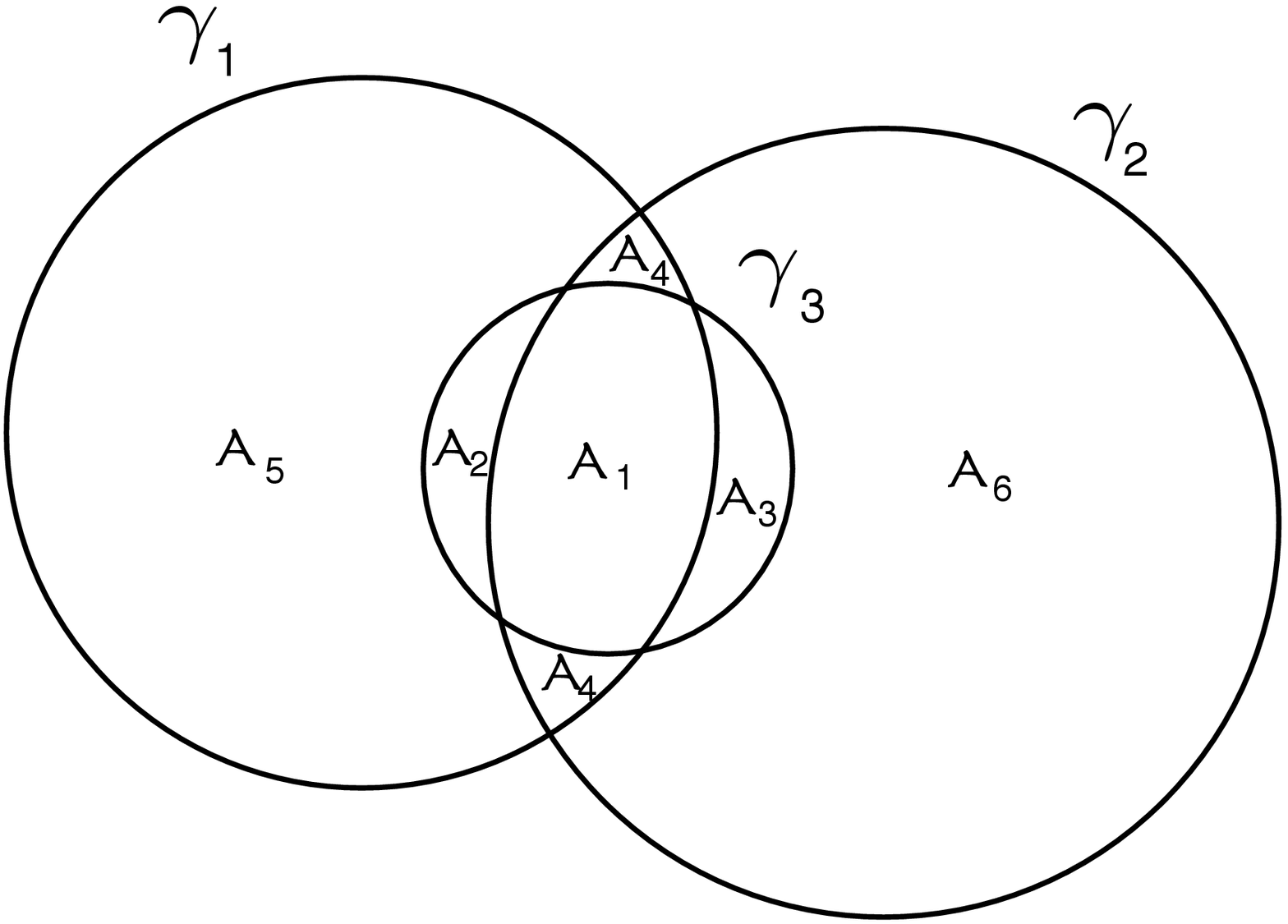}
     \caption{Example of configuration with three circles and identification of the regions of interest.}
  \label{fig:c2}
  \vup\vup\vup
\end{figure}
Let us define the set $\boldsymbol{\mathcal{E}}{=}\bigcup_{n=1,\ldots,N_c}\boldsymbol{\mathcal{E}}^{(n)}$. Then,
the elements of $\boldsymbol{\mathcal{E}}$ are disjoint regions that tessellate
the overall region covered by $\mathcal{C}$.
We also define the measure $\phi$, where $\phi(\mathcal{A})$ corresponds to the area of the region $\mathcal{A}$.

In this paper, we address the problem of computing the measure of the regions $\boldsymbol{\mathcal{E}}^{(n)}$, $n{=}1,\ldots,N_c$,
given the centers and the radii of the circles in $\mathcal{C}$. This appears to be an extremely complex geometrical problem.
In fact, while the area covered by the intersection of two circles has a simple measure, even the intersection area of three circles has a rather
involved form, that depends on the mutual positions of the circles~\cite{trecerchi}. When more than three circles are considered,
the number of configurations grows larger, and the complexity of the geometric conditions and the associated expression
of the intersection area become difficult to handle.

The rest of the paper is organized as follows. Section~\ref{geomres} presents the geometric results that are the foundation of the iterative algorithm.
In Section~\ref{descralg} we describe in detail the structure of the proposed algorithm. In Section~\ref{comparea} we show how the wanted areas can be computed. Finally, in Section~\ref{ex_netdes} we present and discuss the aforementioned network design problem.

\section{Geometric results}
\label{geomres}
In this section, we present the geometric results that represent the core of the proposed algorithm. A key observation is
the following: a necessary (but in general not sufficient) condition for the existence of
the intersection of $n$ circles is the existence of the intersection of all subsets of $n\!-\!1$ circles.\footnote{We say that an intersection \emph{exists} if it is non-empty.} This consideration
may be very useful, since the calculation of several areas among circles which are not all intersecting can be avoided.
Nonetheless, a stronger result can be stated if the number of considered circles is greater than $3$.

\begin{teo}
\label{teo:existn}
Consider a subset of $n{\leq}N_c$ circles $\mathcal{S}{=}\{\gamma_{i_1},\ldots,\gamma_{i_n}\}{\subseteq}\mathcal{C}$ and the associated intersection region $\mathcal{I}^{(n)}_{\{i_1,\ldots,i_n\}}$. With a slight abuse of notation, we refer to the considered subset of circles when denoting the intersection
regions, and we write $\mathcal{I}^{(n-\ell)}_{\{i_1,\ldots,i_n\}\setminus\{e_1,\ldots,e_{\ell}\}}{=}\mathcal{I}^{(n-\ell)}(e_1,\ldots,e_{\ell})$.

Then, $\forall \{i_1,\ldots,i_{n}\}, i_j{\neq} i_u$, for $j{\neq}u$, if $n{\geq}4$ the following holds:
if 
\begin{equation}
\phi(\mathcal{I}^{(n-1)}({k})) > 0 ~~ \forall k\in\{i_1,\ldots,i_n\},
\end{equation}
then also $\phi(\mathcal{I}^{(n)}_{\{i_1,\ldots,i_{n}\}})>0$.
\end{teo}
In other words, the existence of the $n$ intersections among all the subsets of $n-1$ circles in $\mathcal{S}$, besides being an obvious necessary condition for the intersection among the $n$ circles to exist, is also sufficient when $n\geq 4$.

\begin{proof}
Consider a set of $n$ circles $\{\gamma_{i_1},\ldots,\gamma_{i_{n}}\}$. By hypothesis $\phi(\mathcal{I}^{(n-1)}(k)){>}0$, $\forall k{\in}\{i_1,\ldots,i_n\}$.
Fix the index $k{=}\overline{k}$ and define the set $\mathcal{S}_{\overline{k}}{=}\{i_1,\ldots,i_n\}{\setminus }\{\overline{k}\}$. As $\phi(\mathcal{I}^{(n-1)}(\overline{k}))>0$, then $\mathcal{I}^{(n-1)}(\overline{k}){\neq}\emptyset$. Call $\Delta$ the polygon\footnote{$\Delta$ is a particular polygon whose
sides are arcs of circumference.} delimiting $\mathcal{I}^{(n-1)}(\overline{k})$ and
whose sides are $m$ arcs of circumference $\alpha_j$, $j{=}1,\ldots, m$, with $1{\leq} m {\leq} 2(n{-}2)$. Note that more than one arc $\alpha_j$ may
belong to the same circle $\gamma_i$.\footnote{We say that an arc belongs to a circle when it corresponds to a portion of its circumference} We denote with $\boldsymbol{\alpha}(i)$ the set of arcs belonging to the circle $\gamma_i$.

We have to consider three cases:
\begin{enumerate}
 \item $\exists \gamma_i {\in} \mathcal{S}_{\overline{k}}: \; \boldsymbol{\alpha}(i) {=} \emptyset$;
 \item $\forall \gamma_i{\in} \mathcal{S}_{\overline{k}} \; |\boldsymbol{\alpha}(i)| {=} 1$;
 \item $\forall \gamma_i{\in} \mathcal{S}_{\overline{k}} \; \boldsymbol{\alpha}(i) {\neq} \emptyset$ and $\exists j:\; |\boldsymbol{\alpha}(j)| {>} 1$
\end{enumerate}

In the first case $\gamma_i$ fully contains the whole intersection but its circumference does not
hit it. Thus $\mathcal{I}^{(n-1)}(\overline{k}){=}\mathcal{I}^{(n-2)}(\overline{k},i){\cap}\gamma_i{=}\mathcal{I}^{(n-2)}(\overline{k},i)$. Since by assumption $\mathcal{I}^{(n-1)}(k){\neq}\emptyset$, $\forall k{\in}\{i_1,\ldots,i_n\}$,
then also $\mathcal{I}^{(n-1)}(i){=}\mathcal{I}^{(n-2)}(\overline{k},i)\cap\gamma_{\overline{k}}{\neq}\emptyset$. Therefore, $\mathcal{I}^{(n-1)}(\overline{k}){\cap}\gamma_{\overline{k}}{\neq}\emptyset$, and thus $\mathcal{I}^{(n)}_{\{i_1,\ldots,i_n\}}{\neq}\emptyset$ and $\phi(\mathcal{I}^{(n)}_{\{i_1,\ldots,i_n\}}){>}0$.

In the second case, the polygon $\Delta$ is delimited by exactly $n{-}1$ arcs of circumference, each belonging to one of the $n{-}1$ circles in $\mathcal{S}_{\overline{k}}$. This situation is depicted in Figure \ref{fig:c3}. Consider two non-consecutive arcs of $\Delta$, namely $\alpha_r$, $\alpha_t$, and assume, without any loss of generality, that they belong to the circles $\gamma_{i_1}$ and $\gamma_{i_2}$, with $\overline{k}{\neq}i_1$, $i_2$. By assumption, as
the intersection of any combination of $n{-}1$ circles exists, then $\gamma_{\overline{k}}$ must contain at least one point $P {\in} \mathcal{I}^{(n-2)}(i_1,\overline{k})$ and one point $Q {\in} \mathcal{I}^{(n-2)}(i_2,\overline{k})$. Since a circle is a convex figure, then the whole segment joining $P$ and $Q$ must be contained in $\gamma_{\overline{k}}$. Moreover, $P$ and $Q$ both belong to $\mathcal{I}^{(n-3)}(i_1,i_2,\overline{k})$, which is also a convex set (since it is the intersection of convex sets), and therefore the segment $V {=} \overline{PQ} {\subset} \mathcal{I}^{(n-3)}(i_1,i_2,\overline{k})$, and also $V {\subset} \mathcal{I}^{(n-2)}(i_1,i_2)$. Since $P{\in}\gamma_{i_2}$ and $Q{\notin}\gamma_{i_2}$, then $V {\cap} \gamma_{i_2}{\neq}\emptyset$, and analogously $V {\cap} \gamma_{i_1}{\neq}\emptyset$. It follows, from the fact that $V {\subset} \mathcal{I}^{(n-2)}(i_1,i_2)$, that $V$ cannot hit any circumference other than those of $\gamma_{i_1}$ and $\gamma_{i_2}$. In addition, the point of 
intersection between $V$ and the circumference of $\gamma_{i_1}$ must belong to the arc of this circumference contained in $\mathcal{I}^{(n-1)}(i_2)$. This arc is $\alpha_r$, since $\alpha_r$ and $\alpha_t$ are two non consecutive arcs of $\Delta$. Hence, there is one point of $\alpha_r$ which belongs to $V$. As a side of $\Delta$, $\alpha_r{\subset}\mathcal{I}^{(n-1)}(\overline{k})$, whereas $V {\subset} \gamma_{\overline{k}}$: this point of intersection then belongs to $\mathcal{I}^{(n)}_{\{i_1,\ldots,i_n\}}$, which therefore is a non-empty set. The same holds for the point of intersection between $V$ and the circumference of $\gamma_{i_2}$ and, consequently, for the fraction of $V$ between the two intersection points.

Finally, we observe that the found set has non-zero measure by construction,
and therefore $\phi(\mathcal{I}^{(n)}_{\{i_1,\ldots,i_n\}}){>}0$. 
Note that this proof does not hold for $n{=}4$, where all the arcs of $\Delta$ are consecutive. However, in this case the thesis can be proved in an analogous way.
\begin{figure}
    \centering
    \includegraphics[width=\figw]{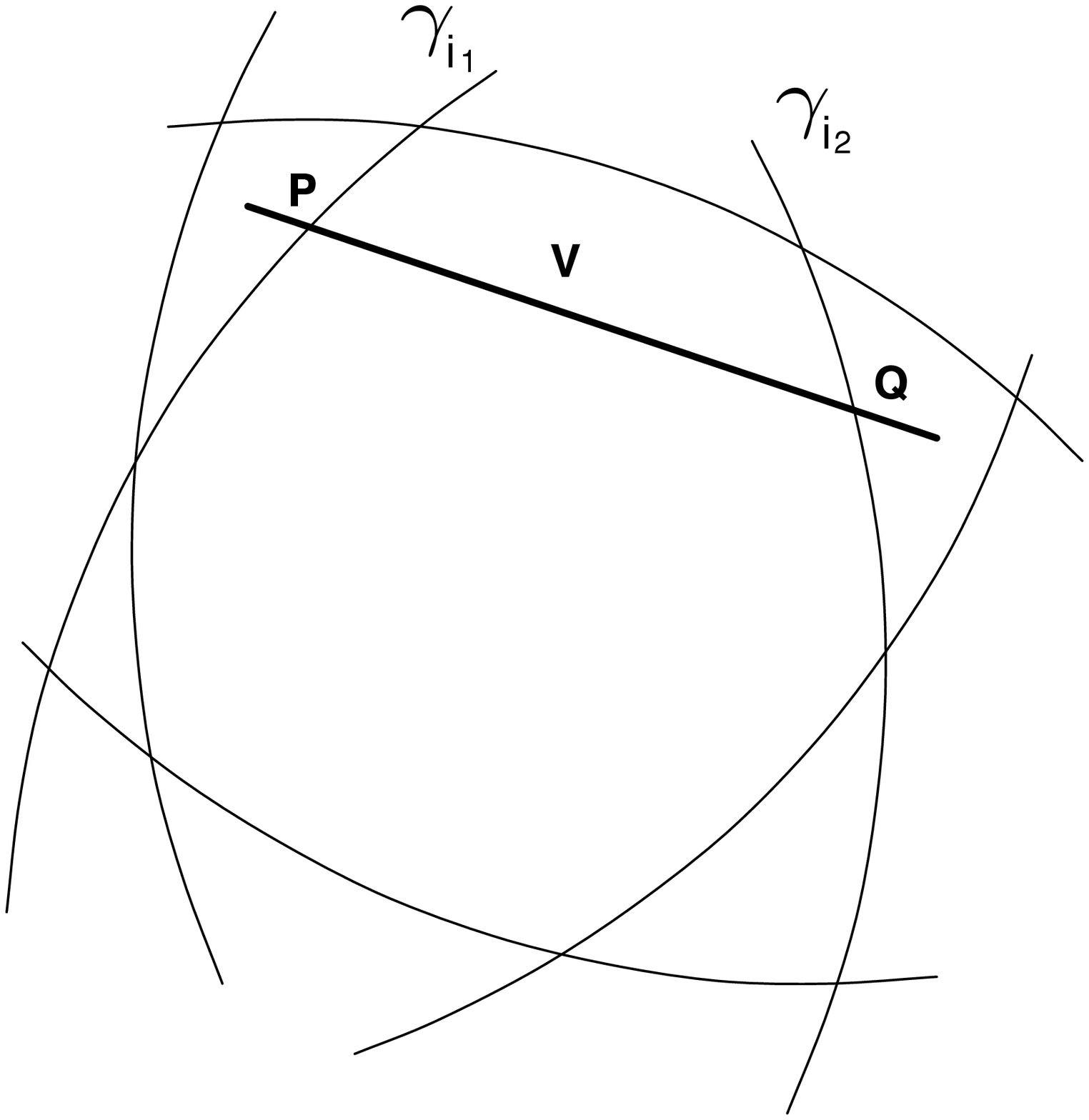}
     \caption{Intersection of $n$ circles. Here each circumference contains one arc belonging to the circular polygon $\Delta$, which delimits the intersection.}
  \label{fig:c3}
  \vup\vup\vup
\end{figure}

In the third case, there is a circle $\gamma_j$ whose associated set of arcs $\boldsymbol{\alpha}(j)$ contains at least two arcs of $\Delta$. It can be shown that two consecutive arcs cannot belong to the same circle. Once observed this, the proof is analogous to that of the previous case by choosing
$\alpha_t$ and $\alpha_r$ as non-consecutive arcs belonging to $\boldsymbol{\alpha}(j)$.
\end{proof}

The previous theorem represents a powerful tool for testing the existence of the intersection regions generated by $n$ circles, once those
generated by $n{-}1$ are known. In fact, the intersection of a set $\mathcal{S}$ of $n{\geq}4$ circles is non-empty if and only if all the intersections
of the subsets of $n{-}1$ circles of $\mathcal{S}$ are non-empty. This simple principle substitutes involved and time-demanding
geometric considerations.

The following theorem states an important property which can be used to derive an effective way to compute the area of any intersection of $n$ circles as a function of some specific intersections generated by $n{-}1$
and $n{-}2$ circles.

\begin{teo}
 \label{teo:intinun}
Consider a subset of $n$ circles $\mathcal{S}{=}\{\gamma_{i_1},\ldots,\gamma_{i_n}\}{\subseteq}\mathcal{C}$, with $4{\leq}n{\leq}N_c$, and the associated intersection region $\mathcal{I}^{(n)}_{\{i_1,\ldots,i_n\}}$, with $\phi(\mathcal{I}^{(n)}_{\{i_1,\ldots,i_n\}}){>}0$. Denote with $m$ the number of arcs of circumference that delimit $\mathcal{I}^{(n)}$.
Then, if $m\geq 4$, there exist two circles $\gamma_t,\gamma_r {\in} \mathcal{S}$ such that:
 \be
 \bigcap_{i\neq t,r} \gamma_i \subset \gamma_t \cup \gamma_r
 \ee
\end{teo}

\begin{proof}
See Appendix \ref{app:proof2}.
\end{proof}

Theorem $\ref{teo:intinun}$ requires $\Delta$ to have a number of arcs greater than $3$. If this does not hold, $\mathcal{I}^{(n)}$ collapses into the intersection of two or three circles, which can be obtained geometrically. If four or more circles are involved, Theorem $\ref{teo:intinun}$ ensures that the area of $\mathcal{I}^{(n)}$ can be found by algebraic manipulation of the intersections among $n{-}1$ and $n{-}2$ circles.

In fact, if Theorem $\ref{teo:intinun}$ holds, the intersection area $\mathcal{I}^{(n-2)}({t,r})$ among all circles but $\gamma_t$ and $\gamma_r$ is fully included in the union between $\gamma_t$ and $\gamma_r$. The area of this intersection can be calculated by considering the partition of $\gamma_t{\cup}\gamma_r$ into three regions, one belonging only to $\gamma_t$, one to both $\gamma_t$ and $\gamma_r$, and the third only to $\gamma_r$. All these three regions exist, the second one by hypothesis, the first and the third one because otherwise the required intersection $\mathcal{I}^{(n)}$ would be fully included in one circle, and the problem could be reduced to an analogous one with a lower number of circles. Call these areas $\mathcal{B}_1$, $\mathcal{B}_2$ and $\mathcal{B}_3$. $\mathcal{I}^{(n-2)}(t,r)$ intersects $\mathcal{B}_2$ by hypothesis, and it may also intersect $\mathcal{B}_1$, $\mathcal{B}_3$, or even both. Assume that it intersects all three areas, and call the intersections respectively $\
mathcal{A}_1$, $\mathcal{A}_2$ and $\mathcal{A}_3$. It is clear that $\mathcal{A}_2 = \mathcal{I}^{(n)}$. These three areas are all unknown. However, it follows from their definition that $\mathcal{A}_1\cup\mathcal{A}_2 = \mathcal{I}^{(n-1)}(t)$, $\mathcal{A}_2\cup\mathcal{A}_3 = \mathcal{I}^{(n-1)}(r)$, and $\mathcal{A}_1\cup\mathcal{A}_2\cup\mathcal{A}_3 = \mathcal{I}^{(n-2)}(r,t)$. Then
\begin{equation}
 \phi(\mathcal{A}_2) = ( \phi(\mathcal{A}_1) {+}  \phi(\mathcal{A}_2)) + (\phi(\mathcal{A}_2) + \phi(\mathcal{A}_3)) - ( \phi(\mathcal{A}_1){+}   \phi(\mathcal{A}_2) + \phi(\mathcal{A}_3)
\end{equation}
and, therefore:
\begin{equation}
 \label{areasucc}
 \phi\left(\mathcal{I}^{(n)}\right) = \phi\left(\mathcal{I}^{(n-1)}(t)\right)\! + \phi\left(\mathcal{I}^{(n-1)}(r)\right) - \phi\left(\mathcal{I}^{(n-2)}(t,r)\right)
\end{equation}
It can be similarly shown that even if $\mathcal{I}^{(n-2)}(t,r)$ does not intersect $\mathcal{B}_1$, $\mathcal{B}_3$ or both, the result still holds.
\begin{figure}
\centering
\includegraphics[width=\figw]{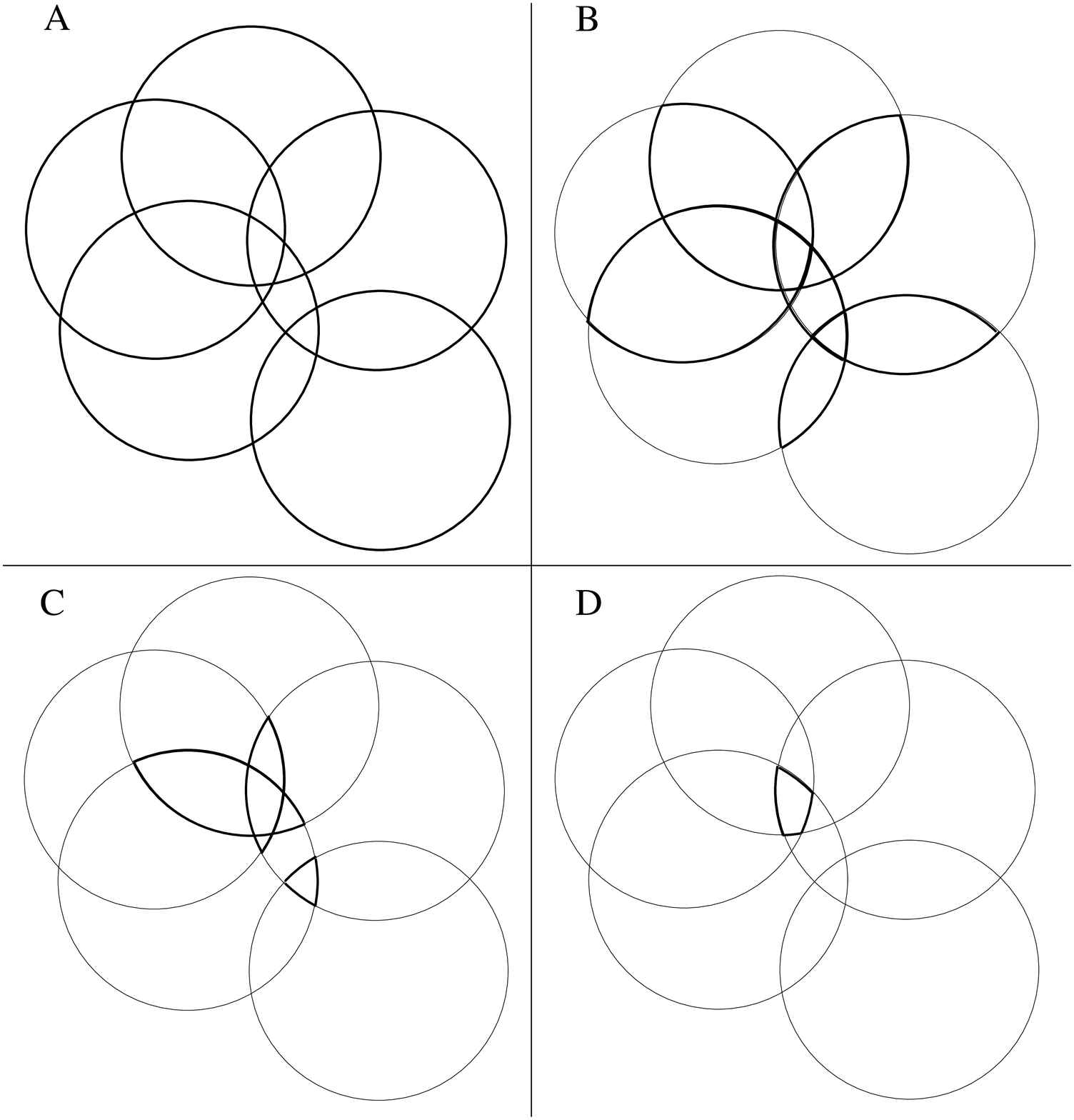}
\caption{Example of circles configuration and areas considered in the first four steps of the algorithm.}
\label{fig:cercalgo}
  \vup\vup\vup
\end{figure}

\section{The Algorithm}
\label{descralg}
In this Section, we describe in detail the proposed algorithm. We remark that the algorithm computes both the non exclusive intersection areas (set $\boldsymbol{\mathcal{I}}$) and the exclusive ones (set $\boldsymbol{\mathcal{E}}$). Since the cardinality of $\boldsymbol{\mathcal{I}}$ grows exponentially with the number of circles, it follows that also the complexity of the algorithm is exponential in the worst case (the cardinality of $\boldsymbol{\mathcal{E}}$ instead increases with the square number of the deployed circles).
Indeed, it is shown in Section \ref{algo_complex} that this happens when a common intersection exists among all the circles, with no nested circles. In this case, the overall complexity grows as $N_c^22^{N_c}$. The algorithm, exploiting the previously derived results, iteratively checks the existence and computes the intersection areas thus achieving a greatly reduced complexity in most cases.

As stated before, the algorithm is based on a trellis structure, which is built iteration by iteration with a simple procedure. After having retrieved all the elements of $\boldsymbol{\mathcal{I}}$, the exclusive intersection areas are also derived and collected in a set of vectors, which is the output of the algorithm.

In this Section, we first describe the trellis structure, showing why it is useful to represent and order the intersection areas, and how it is built through simple transition matrices; secondly, we list the auxiliary variables which are computed at each step of the algorithm, and explain how they are used; finally, we show how the auxiliary variables are updated at each step, and how, at the end of the last step, the vectors containing the exclusive intersection areas can be retrieved from them. 

\subsection{Trellis Structure}
The algorithm is based on a trellis structure. In each iteration, the algorithm takes as input the structure and the variables built at the previous steps,
in order to update the trellis.

Let us clarify how the algorithm works via a graphical example. Consider the circles
in Fig.~\ref{fig:cercalgo}, where the first four steps of the algorithm are depicted in the subfigures A, B, C and D,
respectively. In the first step, the algorithm simply calculates the areas of all the circles. In the second step, the algorithm checks the existence
and computes the intersection areas of any pair of circles via simple geometric calculations. In the third step, the intersection areas
of any existing intersection between any triplet of circles are computed. This is the last step involving geometric computations, as 
in the subsequent steps the existence and the measure of the intersection regions of any set of $4{\leq}n{\leq}N_c$ circles are carried out
via simple algebraic manipulation of the areas computed at the previous steps, exploiting the two theorems presented in Section~\ref{geomres}.

In order to effectively keep track of the existence and value of the regions computed in each step we build a proper graph, in which the various intersection areas correspond to the vertices of the graph. In this way, the relationships can be represented as the edges of the graph. As a byproduct, the graph is also useful to list and order all the areas to be computed.

Consider a set $\mathcal{C}$ of $N_c$ circles, labeled as $\gamma_i$, with $1{\leq} i{\leq} N_c$. Our trellis structure is given by a set of vertices $\mathcal{V}$ and a set of edges $\mathcal{R}$. Each vertex corresponds to an intersection among some of the circles belonging to the set $\mathcal{C}$ considered. Hence, there is a one-to-one relationship among the vertices and all possible subsets of $\mathcal{C}$, except the empty set. We define as $\mathcal{S}(v){\subseteq}\mathcal{C}$ the set of circles whose intersection corresponds to vertex $v$ in the graph.

The vertices are divided into the subsets $\mathcal{V}(n)$, with $n {=} 1{,} 2{,}\ldots, N_c$. The subset $\mathcal{V}(n)$ is the set of all the intersections among $n$ circles out of the $N_c$ considered (that is, the set of all the subsets of $\mathcal{C}$ with exactly $n$ elements).

The obtained subsets can be ordered from $\mathcal{V}(1)$ to $\mathcal{V}(N_c)$. Each edge can connect only two vertices belonging to two consecutive subsets $\mathcal{V}(n)$ and $\mathcal{V}(n{+}1)$. An edge connecting $j{\in}\mathcal{V}(n)$ and $i{\in}\mathcal{V}(n{+}1)$ exists if and only if $\mathcal{S}(j){\subset}\mathcal{S}(i)$. In other words, $i$ is reachable from $j$ if it corresponds to the intersection of all the circles associated with $j$ plus another one.

Finally, also the elements of each $\mathcal{V}(n)$ can be ordered. Vertex $i$ of the graph can be uniquely identified by a binary sequence of $N_c$ elements $\textbf{b}_i {=} [b^i_1,b^i_2,{\ldots},b^i_{N_c}]$ such that $b^i_t {=} 1$ if $\gamma_t{\in}{\mathcal{S}(i)}$. For instance, if $N_c{=}5$, the vertex corresponding to the intersection of $\gamma_2$, $\gamma_3$ and $\gamma_5$ can be labeled as $01101$. An equivalent labeling is obtained using the decimal representations of the binary sequences. In the example above, the same vertex is hence labeled as $13$. The vertices of $\mathcal{V}(n)$ are then ordered with decreasing labels.
An example of the full graph for $N_c{=}5$ is also reported in Fig.~\ref{fig:tre}.

The idea behind the trellis is straightforward: a vertex belonging to $\mathcal{V}(3)$, for instance, corresponds to the intersection of three circles. It follows that this area is related to the three intersection areas among any pair of the three considered circles (represented by three vertices belonging to $\mathcal{V}_2$), since it can be derived from each one of them by adding the missing circle. The exact way through which this can be done is explained below, and relies on Theorem \ref{teo:intinun}.

With the ordering described above, the graph can be fully described by $N_c{-}1$ transition matrices. We define the transition matrix $\textbf{M}_{n,n+1}^{(N_c)}$, for $1{\leq} n{<}N_c$, as a binary matrix containing the information about which edges exist between the vertices in $\mathcal{V}(n)$ and the ones in $\mathcal{V}(n{+}1)$. $\textbf{M}_{n,n+1}^{(N_c)}(i, j) = 1$ if there exists an edge connecting the $i$--th vertex of $\mathcal{V}({n{+}1})$ and the $j$--th vertex of $\mathcal{V}(n)$, and $0$ otherwise.

\begin{figure}
\centering
\includegraphics[width=\figw]{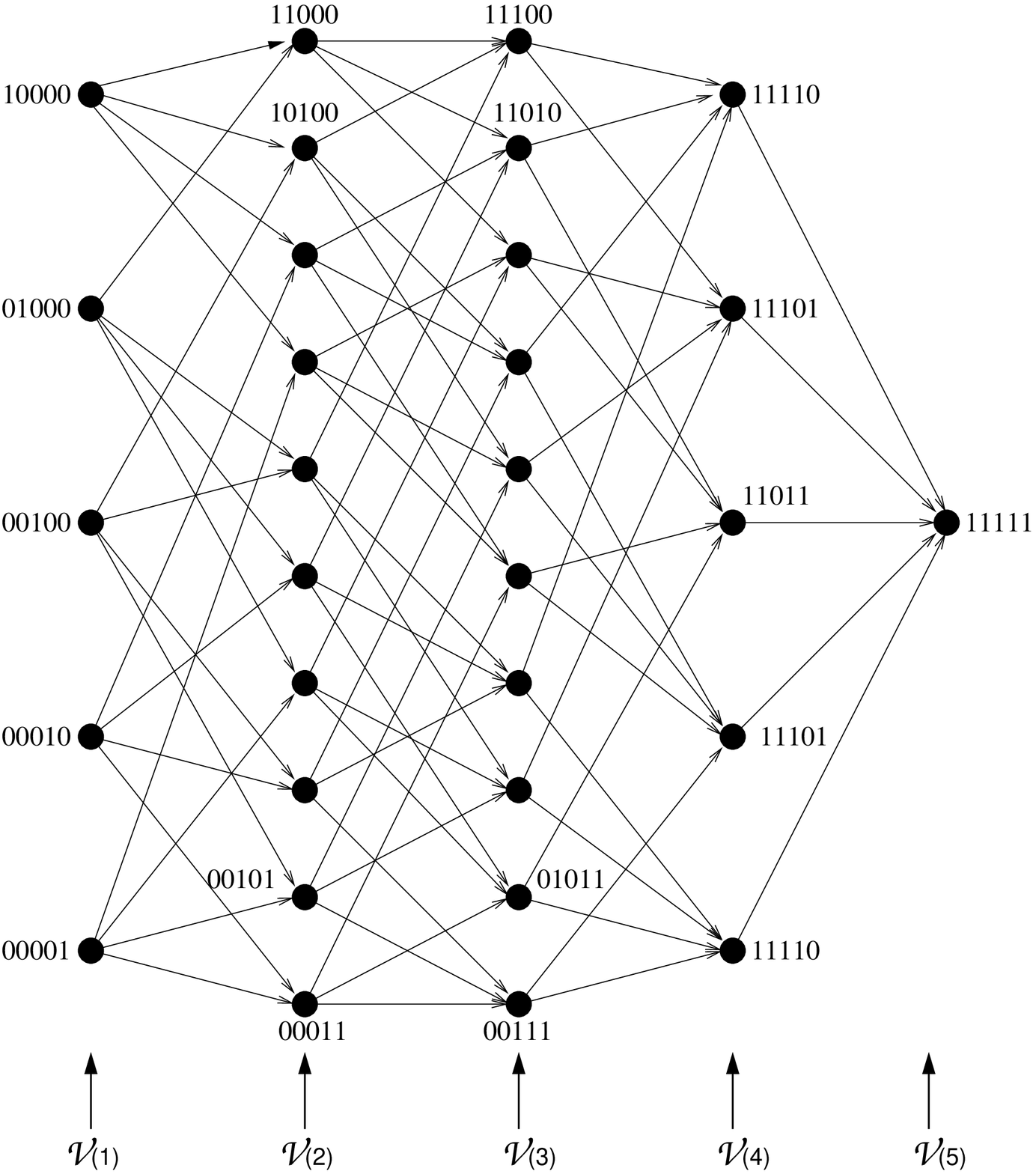}
\caption{The trellis structure for $n=5$. Some of the binary labels are reported. The $i$--th column of vertices corresponds to the subset $\mathcal{V}_i$.}
\label{fig:tre}
  \vup\vup\vup
\end{figure}
The computation of the transition matrices can be performed recursively, if we consider an additional subset $\mathcal{V}(0)$, which contains a single element corresponding to the empty set $\emptyset$, that is, to the region not covered by any circle. This ``virtual'' vertex is not included in the trellis in Fig.~\ref{fig:tre}, though it can be inserted on the left side and connected to all the vertices belonging to $\mathcal{V}(1)$.

With this modification, the following properties about the transition matrices hold:
\begin{itemize}
 \item $\textbf{M}_{n,n+1}^{(N_c)}$, for $1\leq n < N_c$, is a $p\times q$ matrix, where
\be
 p = \binom{N_c}{n+1}, \; q = \binom{N_c}{n}
\ee
 \item $\textbf{M}_{0,1}^{(N_c)}$ is a $N_c\times 1$ column vector whose elements are all equal to 1.
 \item If $n{\geq} {\lceil} N_c{/}2{\rceil}$, then $\textbf{M}_{n,n+1}^{(N_c)} = \left(\textbf{M}_{N_c-n-1,N_c-n}^{(N_c)}\right)^c$, where $^c$ indicates the transposition along the 
 secondary diagonal, i.e., the element of the matrix of the $i$--th row and $j$--th column is moved to the $q{-}j{+}1$--th row and $p{-}i{+}1$--th column.
 \item Each transition matrix can be derived recursively as:
\be
 \textbf{M}_{n,n+1}^{(N_c)} = \left[
 \begin{array}{cc}
 \textbf{M}_{n-1,n}^{(N_c-1)} & \textbf{I} \\
 \textbf{0} & \textbf{M}_{n,n+1}^{(N_c-1)}\\
\end{array}
\right]
\label{recur}
\ee
where $\textbf{I}$ is the identity matrix. We point out that, in any case, it must be $n<N_c$. However, for $n{\geq} {\lceil} N_c{/}2{\rceil}$, due to the third property of the transition matrix, we can compute instead $\textbf{M}_{N_c-n-1,N_c-n}^{(N_c)}$. The recursion appears in the first and in the fourth term of the matrix on the right side in Eq.~(\ref{recur}). In the first term both $N_c$ and $n$ are reduced at each iteration, until we get to a matrix with $n=0$, which can be derived using the second property of the transition matrices. As regards the fourth term, at each recursion only $N_c$ is reduced. In this case, the recursion ends when $N_c-1 = n+1$. In fact, with the exception of the trivial case $n=0$, the matrix $\textbf{M}_{n,n+1}^{(N_c-1)}$ can now be derived from $\textbf{M}_{0,1}^{(N_c-1)}$, due to the third property, which is in turn computed according to the second property.
\end{itemize}

Different transition matrices can be multiplied together, thus giving information about which vertices of a set $\mathcal{V}(n)$ are connected to vertices belonging to $\mathcal{V}(n{+}k)$, with $k{>}1$. The elements of the transition matrices are binary variables, and we define the transition matrix between $\mathcal{V}(n)$ and $\mathcal{V}(n{+}k)$ as
\be
 \textbf{M}_{n,n{+}k}^{(N_c)} = \frac{1}{k!}\prod_{i=0}^{k-1}\textbf{M}_{n+k-i-1,n+k-i}^{(N_c)}
\ee

\subsection{Auxiliary Variables}
\label{sec:auxvar}
As stated before, the algorithm at each step $n$ finds the intersection areas represented by the vertices belonging to $\mathcal{V}(n)$. In order to obtain these areas, it exploits the areas belonging to $\mathcal{V}(n{-}1)$, found in the previous step, and applies Theorem \ref{teo:intinun} based on the relationships expressed by the transition matrix $\textbf{M}_{n-1,n}^{(N_c)}$.\footnote{In the following we drop the superscript of the transition matrices for notation clarity.}  During the subsequent steps, however, some information must be collected and stored. Before describing the steps in detail, we list the auxiliary variables which are computed at each step. At the end of the algorithm, they are used to retrieve the required intersection areas.

\begin{itemize}
 \item We define the $n$--th label vector $\textbf{L}_n$ as the vector containing the decimal labels of all the elements of the set $\mathcal{V}(n)$. Therefore, $\textbf{L}_n$ has length $\displaystyle\binom{N_c}{n}$, and its elements are sorted in decreasing order. These vectors are computed using the transition matrices, since $\textbf{L}_1 = [2^{N_c-1}, 2^{N_c-2}, \hdots, 2, 1]^T$, and
\be
 \textbf{L}_n = \frac{1}{n-1}\,\textbf{M}_{n-1,n}\,\textbf{L}_{n{-}1}
 \label{lsucc}
 \ee
 \item The $n$--th existence vector $\textbf{E}_n$ is a binary vector containing information on the intersection areas of $n$ circles; therefore, its size is equal to the cardinality of $\mathcal{V}(n)$. The $i$--th element of $\textbf{E}_n$ is equal to 1 if the intersection among the circles associated with the $i$--th vertex in $\mathcal{V}(n)$ is not empty. Assuming that all the $N_c$ circles belonging to $\mathcal{C}$ have positive radius, then $\textbf{E}_1$ is an $N_c{\times} 1$ column vector whose elements are all equal to 1. The other vectors can be recursively computed with the following rule:
\be
 \textbf{E}_{n{+}1} = \max\left(\textbf{M}_{n,n+1}\textbf{E}_n {-} n\,\textbf{1},\, \textbf{0}\right)
 \label{condiex}
\ee
where $\textbf{1}$ and $\textbf{0}$ are here column vectors of length $\displaystyle\binom{N_c}{n{+}1}$, with all elements equal to 1 and 0 respectively, and the maximum is taken element-wise. Equation (\ref{condiex}) can be explained as follows. The intersection area $\mathcal{A}_w$ of $n{+}1$ circles, represented by vertex $w{\in}\mathcal{V}(n{+}1)$, can exist only if all the intersection areas among any $n$ of those circles also exist. These areas are in turn represented by $n{+}1$ vertices belonging to $\mathcal{V}(n)$, which are all connected to $w$ in the graph. Assume that $w$ is the $i$--th vertex of $\mathcal{V}(n{+}1)$: the existence of $\mathcal{A}_w$ is indicated by $\textbf{E}_{n{+}1}(i)$. Note now that the multiplication of the $i$--th row of $\textbf{M}_{n,n+1}$ and $\textbf{E}_n$ is equal to the number of existing intersections among $n$ circles out of the $n{+}1$ associated with $w$. Therefore, in order for $\mathcal{A}_w$ to exist, this product must be equal to $n{+}1$. This holds for each 
element of $\textbf{E}_{n{+}1}$, which is then reduced to a binary vector by subtracting $n$ and nulling the negative elements. When we are considering sets of up to two circles we need to do an additional check. In particular,  This means that even when $\textbf{E}_{n{+}1}(i) {=} 1$, the corresponding area may not exist, and the existence vector must therefore be modified accordingly. This check is not needed for sets of three or more circles due to Theorem \ref{teo:existn}.
\item We define the $n$--th area vector $\textbf{A}_n$ as the vector containing the values of all the intersection areas among $n$ circles. As $\textbf{L}_n$ and $\textbf{E}_n$, it is a column vector of $\displaystyle\binom{N_c}{n}$ elements. The values corresponding to non existing areas are equal to 0. We remark that the outcome of these calculations is the set $\boldsymbol{\mathcal{I}}^{(n)}$, $n{=}1,\ldots,N_c$, while our goal is to compute also the exclusive intersection regions, i.e., the set $\boldsymbol{\mathcal{E}}^{(n)}$. We will show how to compute those areas later.
\end{itemize}

Besides those vectors, the algorithm also keeps track of the vectors $\textbf{r}$, $\textbf{x}_c$ and $\textbf{y}_c$, containing the radii and the coordinates of the centers of the $n$ circles. Finally, the symmetric matrix $\textbf{D}$ contains the distances among the centers: $\textbf{D}(i,j)$ is the distance between the centers of the circles $\gamma_i$ and $\gamma_j$.

\subsection{Computation of the Intersection Areas}
\label{comparea}
Based on the trellis structure and the auxiliary variables introduced above, we now illustrate how the algorithm works.

We recall that the aim of the algorithm is to calculate the areas of all the exclusive intersections of the set $\boldsymbol{\mathcal{E}}$. The algorithm uses the auxiliary variables $\textbf{L}_n$ and $\textbf{E}_n$ to efficiently compute the elements of the vectors $\textbf{A}_n$, for $n\in\{1,2,\ldots,N_c\}$.

The algorithm is divided in three phases:
\begin{itemize}
 \item Initialization phase, where all the auxiliary variables are initialized;
 \item Trellis Computation phase, whose aim is to compute the values of all the vectors $\textbf{A}_n$, for $n\in\{1,2,\ldots,N_c\}$;
 \item Data Processing phase, where the exclusive intersections of the set $\boldsymbol{\mathcal{E}}$ are retrieved from the vectors $\textbf{A}_n$.
\end{itemize}

We report in the following how the algorithm works in each phase.

In the initialization phase, the auxiliary variables $\textbf{L}_1$, $\textbf{E}_1$ and $\textbf{A}_1$ are calculated as stated in Section \ref{sec:auxvar}. Also, the transition matrices are recursively derived. We assume that all the $N_c$ circles involved have finite and positive radius, and thus the measure of
the region covered by each circle is strictly positive.

The Trellis Computation phase then is performed in $N_c-1$ subsequent steps. At each step $n$, the aim of the algorithm is to calculate and store the corresponding area vector $\textbf{A}_n$. To this purpose, it performs the following operations:
\begin{itemize}
\item generation of $\textbf{L}_n$ from $\textbf{L}_{n{-}1}$ using (\ref{lsucc}). This variable is useful, since the binary labels, which can be easily obtained from the decimal ones, provide an effective way to recognize which vertices of $\mathcal{V}(n{-}1)$ are connected to each vertex of $\mathcal{V}(n)$;
\item generation of $\textbf{E}_n$ from $\textbf{E}_{n{-}1}$ using (\ref{condiex}). The number and positions of the nonzero elements of $\textbf{E}_n$ provide information on which intersection areas must be calculated. For $n{\leq} 3$, however, a geometric check is necessary for each of these elements, since a nonzero value does not necessarily mean that the intersection represented by the corresponding vertex exists. The check is done based on the centers and the radii of the involved circles (identified through the label vector $\textbf{L}_n$), and it is then possible to obtain $\textbf{E}_n$. No checks are needed for $n{>}3$, thanks to Theorem \ref{teo:existn};
\item calculation of $\textbf{A}_n$ from the already known vectors $\textbf{A}_i$, with $i<n$. If $n\leq 3$, the existing areas, according to the information given by $\textbf{E}_n$, can be calculated geometrically, applying known formulas with the centers and the radii of the intersecting circles, again retrieved through $\textbf{L}_n$. If on the contrary $n{>}3$, then we proceed as follows.
For the $k$--th element of $\textbf{A}_n$, that is, $\textbf{A}_n(k)$, we can consider a reduced version of the trellis, containing only the circles whose intersection is represented by $\textbf{A}_n(k)$. In this reduced trellis, all the vectors $\textbf{A}_i$, for $i < n$, are known, each of them containing a suitable subset of the elements of the vectors of the original trellis. The problem is now equivalent to the more general one of finding $\textbf{A}_{N_c}$ in a trellis when all the other vectors $\textbf{A}_i$ are already known, and we refer to this case from now on.

To solve this problem, we define the vectors $\hat{\textbf{A}}_i$, for $i\in\{1,2,\ldots,N_c-1\}$, as:
\be
 \hat{\textbf{A}}_i {=}\!\begin{cases}
                       \textbf{A}_i {-}\!\!\!\! \displaystyle\sum_{j=i+1}^{N_c-1}\!\!(-1)^{j-i+1}\textbf{M}_{i,j}^T\textbf{A}_j\!\! &i{<}N_c{-}1\\
			\textbf{A}_i &   i{=}{N_c}{-}1
                      \end{cases}
\label{ricompo}
\ee
Now $\textbf{A}_{N_c}$, which is a scalar, can be found from the vectors $\hat{\textbf{A}}_i$. More specifically, we state that, for $N_c>4$, $\textbf{A}_{N_c}$ is equal to the maximum element of $-\hat{\textbf{A}}_{N_c-2}$. This holds also for $N_c = 4$, unless the minimum elements of $\textbf{A}_1$ and $\textbf{A}_3$ are equal, and greater than the maximum element of $-\hat{\textbf{A}}_{N_c-2}$. In this special case, an additional geometric check is required, since the value of $\textbf{A}_{N_c}$ may instead be equal to the minimum element of $\hat{\textbf{A}}_{N_c-1}$. A formal proof of these statements, as well as the required check for the special case $N_c = 4$, is reported in Section \ref{proofcorrect}, and is based on Theorem \ref{teo:intinun}.
\end{itemize}

After the last step of the Trellis Computation phase, all the vectors $\textbf{A}_i$, with $i\in\{1,2,\ldots,N_c\}$, are known. However, the elements of these vectors are not belonging to the set $\boldsymbol{\mathcal{E}}$, since they are the areas of the non exclusive intersections.
The Data Processing phase of the algorithm performs the computation of the exclusive intersection areas, contained in the vectors $\tilde{\textbf{A}}_i$, with $i\in\{1,2,\ldots,N_c\}$.

It can be shown (a sketch of the proof is reported in Section \ref{proofcorrect}) that the following equation holds:
\be
 \tilde{\textbf{A}}_i {=}\!\begin{cases}
                       \textbf{A}_i {-}\!\!\!\! \displaystyle\sum_{j=i+1}^{N_c}\!\!(-1)^{j-i+1}\textbf{M}_{i,j}^T\textbf{A}_j\!\! &i{<}N_c\\
			\textbf{A}_i &   i{=}{N_c}
                      \end{cases}
\label{ricompall}
\ee
where the equality between $\textbf{A}_{N_c}$ and $\tilde{\textbf{A}}_{N_c}$ is intuitive, since the non exclusive and the exclusive intersection among all the $N_c$ circles are necessarily equal. Equation (\ref{ricompall}) is similar to (\ref{ricompo}), except that now the sum is up to $N_c$, since here all the vectors $\textbf{A}_i$ are known. The elements of the vectors $\tilde{\textbf{A}}_i$ are the required areas, and this concludes the algorithm.

\subsection{Proof of algorithm correctness}
\label{proofcorrect}
We prove here that, if the vectors $\textbf{A}_i$ and $\hat{\textbf{A}}_i$ are known, for $1 \leq i < N_c$, then $\textbf{A}_{N_c}$, which is a scalar, can be found as the maximum value of $-\hat{\textbf{A}}_{N_c-2}$. This holds for $N_c > 4$, and very often also for $N_c = 4$ where, however, in some special cases the required value is instead equal to the minimum value of $\hat{\textbf{A}}_{N_c-1}$. Finally, for $N_c \leq 3$, all the areas can be found geometrically. 

We recall that the vectors $\textbf{A}_i$ contain the non exclusive intersection areas, whereas the vectors $\hat{\textbf{A}}_i$ can be obtained using (\ref{ricompo}). The proof is structured as follows:
\begin{itemize}
 \item we first determine an expression for the elements of the vectors $\hat{\textbf{A}}_i$. We focus on the first element of $\hat{\textbf{A}}_1$, since all the other ones can be retrieved in an analogous manner. We express them as sums of exclusive intersection areas.
 \item using the expression of $\hat{\textbf{A}}_i$ and applying Theorem \ref{teo:intinun}, we prove the statement for $N_c > 4$.
 \item we point out in which cases the statement is not true for $N_c = 4$, and determine how the correct value of $\hat{\textbf{A}}_{N_c}$ can be found in these special cases.
\end{itemize}

In the previous Section, it has been stated that the two geometric Theorems \ref{teo:existn} and \ref{teo:intinun} cannot be applied directly in the proposed algorithm. This is because when the algorithm is executed, it is not known a priori how the circles intersect with each other (which would require an exponentially complex conditions check). The only available data at the $i$-th step are the numerical values of the areas computed in the previous $i-1$ steps. More precisely, at step $N_c$, the vectors $\textbf{A}_i$ for $1 \leq i < N_c$ are available. However, the elements of these vectors are not the exclusive intersection areas, as explained before. Vectors $\hat{\textbf{A}}_i$, for $1\leq i < N_c$, can also be computed, according to (\ref{ricompo}), but also these vectors do not contain the values of the exclusive intersection areas, since $\textbf{A}_{N_c}$ is not known, and consequently the sums in (\ref{ricompo}) are up to $N_c-1$.

Therefore, it is not straightforward how the value of $\textbf{A}_{N_c}$ can be obtained starting from the vectors $\textbf{A}_i$ and $\hat{\textbf{A}}_i$, for $1\leq i < N_c$. We focus on $\textbf{A}_{N_c}$ since, in order to compute each element of $\textbf{A}_i$, for $3 < i < N_c$, it is sufficient to run the algorithm while considering only $i$ circles. In this reduced version of the algorithm, $\textbf{A}_i$ has only one element, and its role is exactly the same as $\textbf{A}_{N_c}$ in the non reduced version of the algorithm. Finally, the elements of $\textbf{A}_i$ for $i\leq 3$ can be derived via geometric computation.

\begin{figure}
    \centering
    \includegraphics[width=0.28\textwidth]{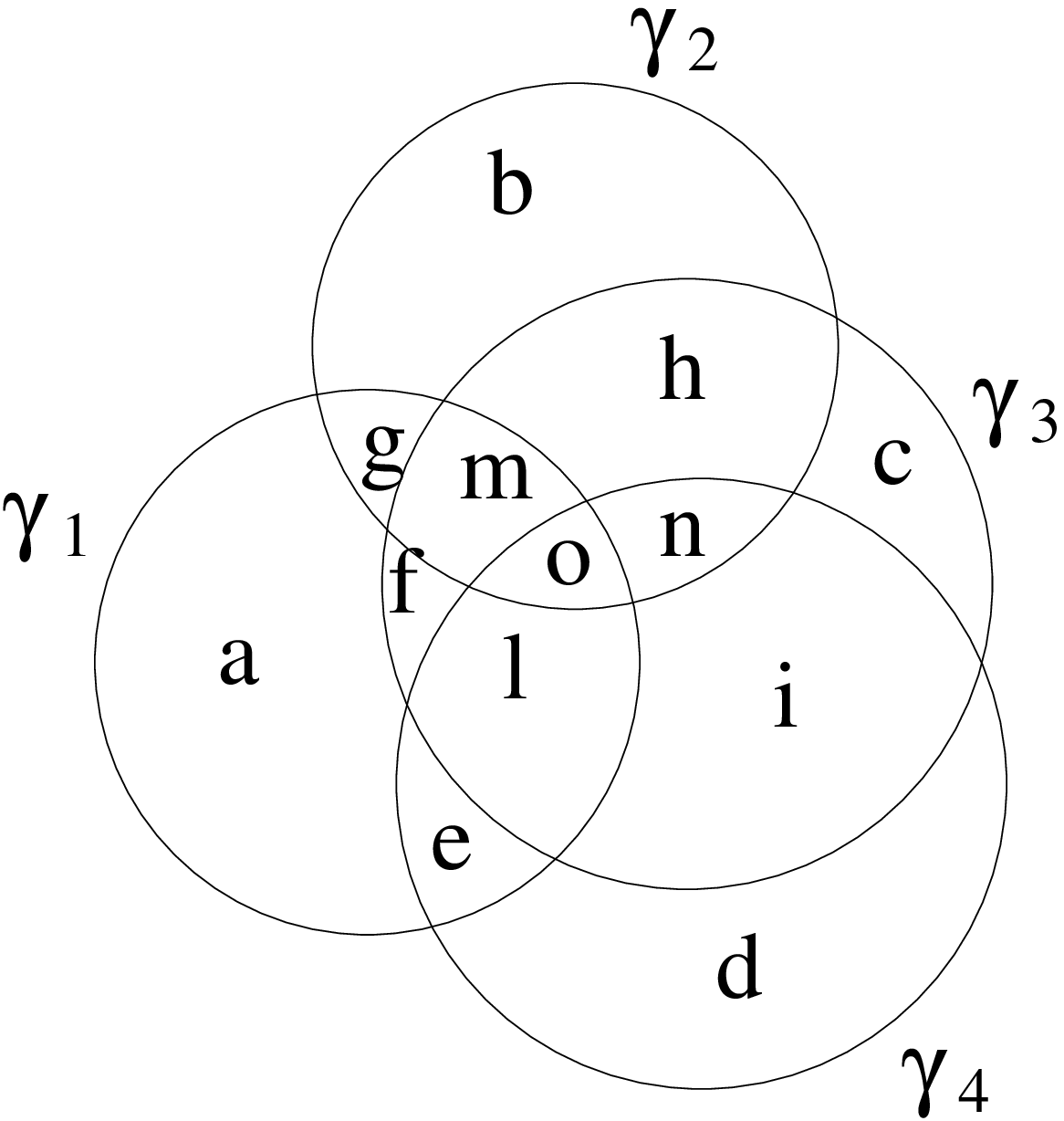}
     \caption{Example of disposition of 4 circles.}
  \label{fig:esecerchi}
  \vspace{-0.5cm}
\end{figure}

The unknown value $\textbf{A}_{N_c}$ is the intersection area of all the $N_c$ circles. In order to explain how this value can be retrieved, we first clarify which are the areas corresponding to the elements of the known vectors. Each element of $\textbf{A}_i$ is the non exclusive intersection area of $i$ circles. In other words, it is the value of the area whose points belong to all the $i$ considered circles, but which may be included also in other circles. In order to distinguish them, we call $\mu_{i_1,i_2,\ldots,i_k}$ the element of $\textbf{A}_k$ corresponding to the non exclusive intersection of circles $\gamma_{i_1}$, $\gamma_{i_2}$..., $\gamma_{i_k}$. In Figure \ref{fig:esecerchi}, we depict an example of deployment of four circles. In this scenario, for example, $\mu_{1,2}$ is given by $g+m+o$, whereas $\mu_{1,3,4}$ is equal to $l+o$. We simply call $\mu$ the intersection of all the $N_c$ circles, which is the only element of $\textbf{A}_{N_c}$, and, thus, the value we want to find. Analogously, we 
call $\mu^*_{i_1,i_2,\ldots,i_k}$ the exclusive intersection of circles $\gamma_{i_1}$, $\gamma_{i_2}$..., $\gamma_{i_k}$. With reference to Figure \ref{fig:esecerchi}, we have $\mu^*_{1,2} = g$, and $\mu^*_{1,3,4} = l$. Note that these values are not necessarily contained in any of the vectors $\textbf{A}_i$ or $\hat{\textbf{A}}_i$.

We now want to express the elements of the vectors $\hat{\textbf{A}}_i$, for $1 \leq i < N_c$ as a function of the exclusive intersection areas, and we start with $\hat{\textbf{A}}_1$, which has exactly $N_c$ elements. The rationale behind equation (\ref{ricompo}) is the following. Since $\textbf{A}_1$ contains the whole areas of all the circles, in order to find the areas covered by exactly one circle we need to subtract the fractions of these areas which are shared with other circles. We first subtract the areas shared by two circles, which are contained in $\textbf{A}_2$. In doing so, since $\textbf{A}_2$ does not contain the exclusive intersection areas, we are subtracting the areas shared by three circles more than once. Hence, we need to sum them again. They are contained in $\textbf{A}_3$. For the same reason, we then need to subtract again the areas shared by four circles, contained in $\textbf{A}_4$, and so forth. If we had the intersection area of all the $N_c$ circles, at the end of this sum we 
would get the exclusive intersection areas of one circle, that is, with a slight abuse of expression, the areas covered by exactly one circle.

For the sake of clarity, we recall again the example of Figure \ref{fig:esecerchi}. The first element of $\textbf{A}_1$ contains the area of $\gamma_1$, which we call $\mu_1$. According to equation (\ref{ricompo}), the first element of $\hat{\textbf{A}}_1$ is computed as:
\bea
 \hat{\textbf{A}}_1(1) & = & \mu_1 - \mu_{1,2} - \mu_{1,3} - \mu_{1,4} + \mu_{1,2,3} + \mu_{1,2,4} + \mu_{1,3,4} \nonumber\\
	& = & (a + e + f + g + l + m + o) - (g + m + o) + \nonumber\\
	& & - (f + l + m + o) - (e + l + o) + \nonumber\\
	& & + (m + o) + (o) + (l + o) \nonumber\\
	& = & a + o
\label{esesomma}
\eea
As can be observed, the result is not the area covered only by circle $\gamma_1$, which would be instead if $\textbf{A}_4 = o$ was subtracted by the result. Since $\textbf{A}_4$ is the area we are looking for, it is necessary to find a way to determine its value.

We first derive an expression for the elements of the vectors $\hat{\textbf{A}}_k$. It is clear from (\ref{ricompo}) that the number of elements of $\hat{\textbf{A}}_k$ is equal to the number of elements of $\textbf{A}_k$ for each $k$. Since the derivation is analogous for every value of $k$, we focus on $\hat{\textbf{A}}_1$.

The first element of this vector is obtained as:
\bea
 \hat{\textbf{A}}_1(1) & = & \mu_1 + \sum_{n=2}^{N_c-1}\left(-1\right)^{n-1}\sum_{1<i_1<1_2<\ldots<i_{n-1}\leq N_c}\mu_{1,i_1,i_2,\ldots,i_{n-1}} \nonumber \\
 & = & \sum_{n=1}^{N_c-1}\left(-1\right)^{n-1}\lambda_n
\label{arcompo}
\eea
where we define
\bea
 \lambda_1 & = & \mu_1; \nonumber \\
 \lambda_n & = & \sum_{1<i_1<i_2<\ldots<i_{n-1}\leq N_c}\mu_{1,i_1,i_2,\ldots,i_{n-1}}
\eea

Since we are computing the first element of $\textbf{A}_1$, we are focusing on circle $\gamma_1$. The first term of the sum is the area of the circle. The second one, $\lambda_2$, is the sum of all the non exclusive intersection areas between $\gamma_1$ and another circle, $\lambda_3$ is the sum of all the non exclusive intersections between $\gamma_1$ and two other circles, and so on.
These areas are not disjoint, since they are the non exclusive intersections. We want to rewrite them in terms of disjoint areas, in order to simplify the expression in (\ref{arcompo}). Define the following sums of disjoint areas:
\bea
 \lambda_1^* & = & \mu_1^* \nonumber\\
 \lambda_n^* & = & \sum_{1<i_1<i_2<\ldots<i_{n-1}\leq N_c}\mu^*_{1,i_1,i_2,\ldots,i_{n-1}}
\eea

According to the above definitions, $\lambda_1^*$ is the area covered only by $\gamma_1$, $\lambda_2^*$ is the sum of all the areas covered only by $\gamma_1$ and another circle, that is, the sum of all the exclusive intersection areas between $\gamma_1$ and another circle. With reference to Figure \ref{fig:esecerchi}, we can write for instance:
\bea
 \lambda_1 = a + e + f + g + l + m + o & , & \lambda_1^* = a, \nonumber\\
 \lambda_2 = e + f + g + 2l + 2m + 3o & , & \lambda_2^* = e + f + g. \nonumber
\eea

We can now express the non exclusive intersection areas, which appear in (\ref{arcompo}), as functions of the exclusive ones, by writing relationships between the $\lambda_i$'s and the $\lambda^*_i$s. In order to do so, we note the following facts:
\begin{itemize}
 \item $\lambda_1$ is the area of $\gamma_1$, and is hence given by the sum of all the exclusive intersection areas involving $\gamma_1$, that is:
\bea
 \lambda_1 & = & \sum_{i=1}^{N_c}\lambda^*_i
\eea
 \item $\lambda_2$ is the sum of all the non exclusive intersection areas between $\gamma_1$ and another circle. In doing this sum, the areas covered by more than two circles (one of which is, by definition, $\gamma_1$), are counted more than once. More precisely, the areas covered by $n$ circles are counted $n-1$ times, and hence:
\bea
 \lambda_2 & = & \sum_{i=2}^{N_c}\left(i-1\right)\lambda^*_i
\eea
 \item in general, when computing $\lambda_k$, the areas covered by $n \geq k$ circles are counted as many times as the number of possible extractions of $k-1$ circles out of $n-1$. On the contrary, the areas covered by $n < k$ circles are never counted. Therefore:
\bea
 \lambda_k & = & \sum_{n=k}^{N_c}\binom{n-1}{k-1}\lambda^*_n
\eea
\end{itemize}

We can order the expressions of the $\lambda_i$ in an $N_c-1 \times N_c$ matrix $\textbf{L}$, whose $i$--th row contains the terms of the sum defining $\lambda_i$. The sign of the even rows is then changed, to get:
\bea
 \textbf{L} & = & \left[
\begin{array}{ccccc}
 \binom{0}{0}\lambda^*_1 & \binom{1}{0}\lambda^*_2 & \binom{2}{0}\lambda^*_3 & \ldots & \binom{N_c-1}{0}\lambda^*_{N_c} \\
  0 & -\binom{1}{1}\lambda^*_2 & -\binom{2}{1}\lambda^*_3 & \ldots & -\binom{N_c-1}{1}\lambda^*_{N_c} \\
  0 & 0 & \binom{2}{2}\lambda^*_3 & \ldots & \binom{N_c-1}{2}\lambda^*_{N_c} \\
  \vdots & \vdots & \vdots & \ddots & \vdots \\
  0 & 0 & 0 & \ldots & (-1)^{N_c}\binom{N_c-1}{N_c-2}\lambda^*_{N_c} \\
 \end{array}
 \!\!\!\right]
\eea

With this representation, according to (\ref{arcompo}), the first element of $\hat{\textbf{A}}_1$ is simply given by the sum of all the elements of $\textbf{L}$. However, while in (\ref{arcompo}) the sum is computed row by row, we now observe that it is simpler to consider the columns. The sum of the elements of the first column is clearly equal to $\lambda^*_1$. For $2\leq n\leq N_c-1$, the sum $s_n$ of the elements of column $n$ is expressed as:
\bea
 s_n & = & \sum_{k=0}^{n-1}(-1)^k\binom{n-1}{k}\lambda^*_n = 0 
\eea
The sum of all the columns is hence 0, with the exception of the last one, due to the fact that the sum in (\ref{ricompo}) is up to $N_c-1$. The calculation of $s_{N_c}$ is straightforward:
\bea
 s_{N_c} & = & \sum_{k=0}^{N_c-2}(-1)^k\binom{N_c-1}{k}\lambda^*_{N_c} \nonumber \\
 & = & \sum_{k=0}^{N_c-1}(-1)^k\binom{N_c-1}{k}\lambda^*_{N_c} - (-1)^{N_c-1}\lambda^*_{N_c} \nonumber \\
 & = & (-1)^{N_c}\lambda^*_{N_c}
\eea

Putting all together, we have that
\bea
 \hat{\textbf{A}}_1(1) & = & \sum_{n=1}^{N_c}s_n \nonumber \\
 & = & s_1 + s_{N_c} \nonumber \\
 & = & \lambda^*_1 + (-1)^{N_c}\lambda^*_{N_c} \nonumber \\
 & = & \mu^*_1 + (-1)^{N_c}\mu^*_{1,2,3,\ldots,N_c} \nonumber \\
 & = & \mu^*_1 + (-1)^{N_c}\mu
 \label{eleacap}
\eea
where $\mu$, as defined above, is the intersection area of all the $N_c$ circles. It is clear that the same derivation can be done for all the elements of $\hat{\textbf{A}}_1$, such that $\hat{\textbf{A}}_1(k) = \mu^*_k + (-1)^{N_c}\mu$. It can also be shown that an analogous procedure can be followed to determine the elements of the other vectors $\hat{\textbf{A}}_n$, with $2\leq n < N_c$. In this case, for each element $\hat{\textbf{A}}_n(k)$ the sums $\lambda_i$ and $\lambda^*_i$ can be defined in the same manner, for $n \leq i \leq N_c$, and an $(N_c-n-1) \times (N_c-n)$ matrix can be constructed, finally resulting in
\bea
 \hat{\textbf{A}}_n(k) & = & \mu^*_{i_1,i_2,\ldots,i_n} + (-1)^{N_c-n+1}\mu
 \label{eleacap2}
\eea
In the expression above, the indices $i_1,i_2,\ldots,i_n$ depend on the element which is being calculated. Reporting again the example in Figure \ref{fig:esecerchi}, we have:
\bea
 \hat{\textbf{A}}_1 = \left[
\begin{array}{c}
 \mu^*_1 + \mu \\
 \mu^*_2 + \mu \\
 \mu^*_3 + \mu \\
 \mu^*_4 + \mu \\
\end{array}
\right],
\hat{\textbf{A}}_2 = \left[
\begin{array}{c}
 \mu^*_{1,2} - \mu \\
 \mu^*_{1,3} - \mu \\
 \mu^*_{1,4} - \mu \\
 \mu^*_{2,3} - \mu \\
 \mu^*_{2,4} - \mu \\
 \mu^*_{3,4} - \mu \\
\end{array}
\right],
\hat{\textbf{A}}_3 = \left[
\begin{array}{c}
 \mu^*_{1,2,3} + \mu \\
 \mu^*_{1,2,4} + \mu \\
 \mu^*_{1,3,4} + \mu \\
 \mu^*_{2,3,4} + \mu \\
\end{array}
\right]
\label{varia}
\eea

The vectors above can be computed by the algorithm at each step. We stress again that the area $\mu$, as well as all the exclusive intersection areas, are unknown, and cannot in general be retrieved from the vectors $\hat{\textbf{A}}_i$. This is true for $N_c > 3$, as stated before, since otherwise all the intersection areas can be found geometrically. We skip for now the special case $N_c = 4$. For $N_c \geq 5$, we can now exploit Theorem \ref{teo:intinun}. The theorem holds only if the number $m$ of circular arcs that delimit the intersection area of all the circles (in this case equal to $\mu$) is greater than or equal to 4.

If this is not true, then there is a circle $\gamma_k$ that fully contains the intersection area of all the other $N_c-1$ circles. This area also belongs to $\gamma_k$, which implies that there is at least one exclusive intersection area among $N_c-1$ circles which is empty. Looking at the example for $N_c = 4$, whose vectors are reported in (\ref{varia}), this means that one of the values $\mu^*_{i,j,k}$ in $\hat{\textbf{A}}_3$ is zero. As a consequence, in order to retrieve $\mu$ it is sufficient in general to take the minimum element of $\hat{\textbf{A}}_{N_c-1}$.

If on the contrary the hypothesis of Theorem \ref{teo:intinun} holds, then it means that there are two circles $\gamma_j$ and $\gamma_k$ that fully contain the intersection area of all the other $N_c-2$ circles. Following the same reasoning as above, it can be concluded that there exists at least one exclusive intersection area among $N_c-2$ circles which is empty, and the value of $\mu$ is the maximum of $-\hat{\textbf{A}}_{N_c-2}$. 

Unfortunately, this is not enough, since it is not known a priori whether the hypothesis of Theorem \ref{teo:intinun} holds or not. If it does not, however, we can use the following argument. If $m = p$, with $1\leq p\leq 3$, there are $p$ circles whose intersection is fully included in any other circle (for $p=1$, there is one circle which is contained in any other circle). We call $\mathcal{P}$ the set of these circles, whereas $\mathcal{Q}$ is the set of the remaining $N_c-p$ circles. Consider $N_c - p - 2$ other circles belonging to $\mathcal{Q}$. The intersection of these circles and the ones belonging to $\mathcal{P}$ is fully contained in the remaining two circles of $\mathcal{Q}$. As a consequence, at least one of the exclusive intersection areas among $N_c-2$ circles is empty, and again $\mu$ can be obtained as the maximum of $-\hat{\textbf{A}}_{N_c-2}$, as in the case where the hypothesis of Theorem \ref{teo:intinun} holds. The abovementioned consideration clearly requires that $N_c$ is at least 
equal to 5, otherwise, if $p=3$, it is not possible to take $N_c-p-2$ circles from $\mathcal{Q}$ \footnote{if $N_c = 5$ and $p = 3$, the intersection of the circles belonging to $\mathcal{P}$ is fully included in the two circles belonging to $\mathcal{Q}$, and the same reasoning still holds.}.

The only case to be studied separately is $N_c = 4$, which is analyzed in depth in Appendix \ref{app:specialcase}.

\section{Algorithm complexity}
\label{algo_complex}
The computation of the algorithm complexity is not straightforward, given its strong dependence on the specific circles deployment. We therefore proceed to determine an upper bound. We first observe that the complexity not only depends on the number of circles $N_c$ considered, but rather on how these circles are deployed. For instance, the time needed by the algorithm when the $N_c$ circles are not intersecting at all grows as $N_c^2$ (linearly for the calculation of the areas, quadratically to verify the absence of any intersection). However, the complexity grows when more circles are intersecting.

We now make the following observation. The algorithm gives, as a result, both the non exclusive and the exclusive intersection areas of $N_c$ circles. Both these results can be useful, depending on the considered application. The most demanding task regards the non exclusive intersection areas, whose number can rise up to $2^{N_c}$. This number is reached when an intersection among all the $N_c$ circles exists.
If this is not true, only a subset of the areas represented by the vertices of the graph reported in Fig. \ref{fig:tre} must be computed. Due to Theorem \ref{teo:existn}, this subset is immediately identified after the calculation of the non exclusive intersection areas among triplets of circles.
In the following we then consider the case where the intersection between all the $N_c$ circles exists.

As regards the exclusive intersection areas, instead, their number is no greater than the number of disjoint areas in which the plane is divided when the circles are deployed. It can be observed that this number can be at most equal to $N_c^2-N_c+2$, which happens when every circumference intersects all the remaining circumferences\footnote{This can be easily proved by adding a circle at a time.
The first circle creates two disjoint areas. When the $i$--th circumference is added, it can intersect at most all the other $i-1$ ones already deployed, each one in at most 2 different points. These $2(i-1)$ points divide the added circumference into $2(i-1)$ arcs. Since all the $i-1$ already deployed circumferences intersect each other by hypothesis, each of these arcs divides an existing intersection area into 2 areas, hence creating a total of $2(i-1)$ new areas. It follows that, when all the $N_c$ circles have been deployed, the total number of exclusive intersection areas is given by $2+\sum_{i=2}^{N_c}2(i-1) = 
N_c^2-N_c+2$.}.

The worst case therefore occurs when there exists an intersection among all the $N_c$ circles and no circles are fully included in other ones. We will focus on this case from now on.

We start with the computation of the non exclusive intersection areas. Having already computed the ones created by couples and triplets of circles, we still have to obtain those among $4,5,\ldots, k$ circles.
We observe from equation (\ref{ricompo}) that in order to find each element of the vector $\textbf{A}_i$, with $4\leq i\leq k$ we need to compute the maximum between a subset of the elements of $\hat{\textbf{A}}_{i-2}$, which is in turn obtained as $\textbf{A}_{i-2} - \textbf{M}_{i-2,i-1}^T\textbf{A}_{i-1}$.
All the matrices $\textbf{M}_{i,j}^T$ are $\binom{N_c}{i}\times\binom{N_c}{j}$ binary matrices, with exactly $\binom{N_c-i}{j-i}$ non-zero elements per row. Therefore, they can be more efficiently replaced by lists of indices, each containing $\binom{N_c}{i}\times\binom{N_c-i}{j-i}$ indices. These lists should be precomputed, possibly in a smart manner, exploiting the recursive formulation expressed in (\ref{recur}), and the symmetry between $\textbf{M}_{n,n+1}$ and $\textbf{M}_{N_c-n-1,N_c-n}$.

With these matrices available, we focus on the number of operations needed to find the vector $\textbf{A}_i$, with $4\leq i\leq k$. The vector has $\binom{N_c}{i}$ elements. The vector $\hat{\textbf{A}}_{i-2}$ has $\binom{N_c}{i-2}$ elements. The computation of each one, following the expression $\textbf{A}_{i-2} - \textbf{M}_{i-2,i-1}^T\textbf{A}_{i-1}$, requires $N_c-i$ subtractions, since this is the number of non--zero elements in each row of $\textbf{M}_{i-2,i-1}^T$. Summing over all the values of $i$, the total number of subtractions is
\begin{equation}
 \sum_{i=2}^{N_c-2}\binom{N_c}{i}(N_c-i) = \frac{1}{2}N_c\left(2^{N_c}-2N_c-2\right)
\end{equation}

Having derived the vector $\hat{\textbf{A}}_{i-2}$, each element of $\textbf{A}_i$ is found as the minimum among a subset of $\binom{i}{2}$ elements of $\hat{\textbf{A}}_{i-2}$. Assuming that finding the minimum between $n$ elements requires $n$ comparisons, the overall number of comparisons needed is:
\begin{equation}
 \sum_{i=4}^{N_c}\binom{N_c}{i}\binom{i}{2} = \frac{1}{8}\left(2^{N_c}-4N_c+4\right)(N_c-1)N_c
\end{equation}

We have shown that, in order to compute the $2^{N_c}$ non exclusive intersection areas of a family of $N_c$ circles, the number of subtractions needed grows as $N_c2^{N_c}$, while the number of comparisons needed grows as $N_c^22^{N_c}$.
However, in a smart implementation, once an element of $\textbf{A}_i$ is found, it may also be compared with the $i$ elements of $\textbf{A}_{i-1}$ which correspond to the intersections of all but 1 of the $i$ circles.
In fact, whenever $\mathcal{I}^{(i)}_{\{i_1,i_2,\ldots,i_i\}}$ is equal to $\mathcal{I}^{(i-1)}_{\{i_1,i_2,\ldots,i_{i-1}\}}$, it follows that the circle $\gamma_{i_i}$ fully contains the intersection of the other $i-1$ circles, and therefore all the exclusive intersections between subsets of circles containing $\gamma_{i_1},\gamma_{i_2},\ldots\gamma_{i_{i-1}}$ and not containing $\gamma_{i_i}$ can be immediately set to zero. In this manner, at the end of the computation of the non exclusive intersection areas, several of the non existing exclusive intersection areas have been also identified.

Having computed all the non exclusive intersection areas, the last step is to compute the exclusive ones.
Instead of using equation (\ref{ricompall}), the following relationship is also valid:
\begin{equation}
 \tilde{\textbf{A}}_i {=}\!\begin{cases}
                       \textbf{A}_i -\displaystyle\sum_{j=i+1}^{N_c}\textbf{M}_{i,j}^T\tilde{\textbf{A}}_j &i{<}N_c\\
			\textbf{A}_i &   i{=}{N_c}
                      \end{cases}
\end{equation}
which can be used recursively, from $\tilde{\textbf{A}}_{N_c}$ back to $\tilde{\textbf{A}}_1$. Although it seems that the same number of operations is involved, this is not actually true. This is due to the fact that the number of exclusive intersection areas is much lower, as said above.
Starting from the exclusive intersection of all the $N_c$ circles, which is known (being equal to the non exclusive one), the intersections among $N_c-1$ circles are retrieved by means of matrix $\textbf{M}_{N_c-1,N_c}^T$ through the equation defined above. However, some of these areas may result equal to zero. This information can be immediately exploited, by deleting the corresponding elements in the matrix $\textbf{M}_{N_c-2,N_c-1}^T$ (which can be written in the form of a list, as explained above). The result is that when computing the intersection areas among $N_c-2$ circles, only the required operations are performed, neglecting the subtractions of non existing areas.

Now, in order to compute the number of required operations, we refer to the symmetric case where all the $N_c$ circles have the same radius, and when the intersection between all of them is bounded by a circular polygon with $N_c$ sides, each belonging to a different circle. It is easy to verify that by changing the positions or the radii of the circles, the number of exclusive intersection areas cannot increase any more\footnote{This number actually does not change, as long as the hypothesis about the intersection area among all the $N_c$ circles holds, except for cases where three or more circles intersect in a single point.
However, these cases lead to a lower number of exclusive intersection areas; in addition, if the circles are randomly deployed, the probability of these configurations is 0.}.

In the configuration taken into account, it can be proved that each non exclusive intersection between $k$ circles contains $\binom{N_c-k+1}{2}+1$ exclusive intersection areas. The corresponding exclusive intersection area, therefore, can be computed through $\binom{N_c-k+1}{2}$ subtractions. The overall number of operations, in order to find all the exclusive intersection areas, is given by:
\begin{equation}
 \sum_{k=1}^{N_c-1}\binom{N_c}{k}\binom{N_c-k+1}{2} = \frac{1}{8}\left(2^{N_c}N_c(N_c+3) -4N_c(N_c+1)\right)
\end{equation}
which grows as fast as $N_c^22^{N_c}$.

We conclude that, in the worst case, where all the $N_c$ circles intersect each other and there are no nested circles, the complexity of the algorithm, in order to find all the $2^{N_c}$ non exclusive intersection areas and all the $N_c^2-N_c+2$ exclusive intersection ones, grows as $2^{N_c}N_c^2$.

\subsection{Comparison with Monte-Carlo approximation}
Other numerical methods can be used to approximate the values of the intersection areas. One of the most common ones is the Monte-Carlo approximation: a number $N_p$ of points is randomly chosen in the area where the circles are deployed. For each of them, the distances from the $N_c$ centers are computed, thus finding the intersection area it belongs to.
Since the probability of choosing a point in a given intersection area is equal to the ratio between this area and the total deployment area, the values of the intersection areas can be approximated using a high enough number of test points. 
Although the Monte--Carlo method is simpler to implement, we notice that $N_c$ distance computations and comparisons are needed for each deployed point. In addition, the precision of this method depends on the number of deployed points, especially if the ratio between the average circle radius and the side of the deployment area is small.
If this is the case, even checking the existence of the intersection areas between small circles may require a huge number of points. Indeed, for an area which is $p$ times smaller than the total deployment area, in order to find a value with a relative error equal to $\epsilon$ with probability $\Lambda$, the number $N_p$ of points needed can be approximated as follows.

When $N_p$ points are tested, each of them is within the desired area with probability $p$. Therefore, the total number of points which falls within the desired area is a binomial random variable, with mean $pN_p$. If $N_p$ is high enough, which is true in our case, then the binomial random variable can be well approximated with a Gaussian random variable, with mean $pN_p$ and variance $\sigma_x^2 = pN_p(1-p)$. Call this variable $X$. The condition is then:
\begin{equation}
 \mathbb{P}\left[(1-\epsilon)pN_p\leq X\leq (1+\epsilon)pN_p\right] = \Lambda
\end{equation}
By introducing $Y = X/\sigma_x - pN_p$, which is therefore a Gaussian random variable with zero mean and unit variance, we can rewrite the condition as:
\begin{equation}
 \mathbb{P}\left[-\frac{\epsilon pN_p}{\sqrt{pN_p(1-p)}}\leq Y\leq \frac{\epsilon pN_p}{\sqrt{pN_p(1-p)}}\right] = \Lambda
\end{equation}
from which, given the simmetry of the Gaussian pdf, we get:
\begin{equation}
 Q\left(\frac{\epsilon pN_p}{\sqrt{pN_p(1-p)}}\right) = \frac{1-\Lambda}{2}
\end{equation}
where $Q(\cdot)$ is the Gaussian Complementary Cumulative Distribution Function. Solving for $N_p$, we finally obtain the result:
\begin{equation}
 N_p = \frac{1-p}{\epsilon^2p}\left[Q^{-1}\left(\frac{1-\Lambda}{2}\right)\right]^2
\end{equation}

For $p = 0.1$, $\epsilon = 0.01$ and $\Lambda = 0.9$, the required points are almost 250000. The proposed algorithm, on the contrary, gives the exact values of all the non exclusive and exclusive intersection areas, even when the deployment area is much larger than the average circle area (in this case, indeed, the probability that all the circles are intersecting with each other is quite small, which further reduces the algorithm computational burden).

\section{Network Design Application}
\label{ex_netdes}
The algorithm presented in this paper may become a useful tool to determine distribution functions which would be hard to derive analytically. In this Section we show how this can help in network design problems, in order to determine the optimal setting of specific parameters.

Consider a wireless network where fixed Access Points are distributed in a given area, and a mobile terminal whose position is randomly chosen in the same area, with uniform distribution. Assume that the transmitted power of the mobile terminal is fixed and equal to $P_M$, and that a target SNR $\Gamma$ is required for decoding. With a commonly used approximation, consider a circular area around each access point as its coverage area. The radius of this area can be chosen based on the average SNR or on the outage probability. Once the channel model is determined, the radius depends only on the transmission power $P_M$.

It is clear that, if the power is high enough, there are regions covered by two or more access points. If we assume some sort of cooperation among the access points, users that are located in these areas can take advantage of the spatial diversity by transmitting to multiple access points. The network topology has a strong impact on the overall performance: if the considered area is fixed, increasing the density of access points makes it easier for a transmission to be decoded by several receivers, but also implies an increased network deployment cost.

In the following example, we assume flat Rayleigh fading, such that, if $W$ is the noise power, the SNR at distance $d$ from the receiver is given by the well known equation:
\be
 SNR(d) = \frac{P_M\sigma}{AW}d^{-\alpha}\left|h\right|^2
\ee
where $A$ is a pathloss factor, $\alpha$ is the pathloss exponent, and $h$ is the channel fading gain, distributed as a complex Gaussian random variable with zero mean and unit variance. The attenuation factor $\sigma$, due to shadowing effects, is here considered constant.
The radius $R$ of the coverage circles is defined as the distance at which the average SNR at the access point is equal to a given value $\Delta$. Therefore,
\be
 R = \left(\frac{P_M\sigma}{\Delta A W}\right)^{\frac{1}{\alpha}}
\ee
Alternatively, the radius can be determined based on a given outage probability, with no substantial difference in the results. Note also that, since fixed transmit power is assumed, the same analysis could be done for the downlink as well, with Base Stations cooperating in the transmission phase. Usually, however, the coverage bottleneck lies in the uplink, due to the reduced available power at the mobile terminals, thus we focus our attention on this scenario.

\begin{table}
  \caption{System parameters for the cellular-like topology.}
  \centering
  \begin{tabular}{|l|c|}
  \hline
    Noise power $W$ & $-103$dBm\\
    Path-loss exponent $\alpha$ & $3$\\
    Shadowing margin $\sigma$ & $10$ dB\\
    Fixed attenuation parameter $A$ & $30$ dB\\
    SNR threshold $\Delta$ for decoding radius & $10$ dB\\
    SNR threshold $\Gamma$ for outage probability & $10$ dB\\
    Quantization step $\rho$ & $R/50$ m\\
    \hline
  \end{tabular}
  \label{parametri2}
\vup\vup\vup
\end{table}
In this Section, we set the parameters as in Table \ref{parametri2}.
To analyze the performance, we compute the outage probability in the absence of interference, defined as the probability that the SNR is below the target value $\Gamma$. If the node is in the coverage area of two or more access points, we assume that a Maximum Ratio Combining strategy is applied. A very robust wired channel among the access points is assumed, and the fading is considered independent among different channels.
Therefore, the SNR of the transmission from a source to $n$ access points can be simply written as
\be
 SNR(d_1,d_2,\ldots,d_n) = \frac{P_M\sigma}{AW}\sum_{i=1}^nd_i^{-\alpha}\left|h_i\right|^2
 \label{distSNR}
\ee
which depends on the distances from all the access points within coverage range. The circumferences around the positions of the access points partition the deployment area $\mathcal{A}_D$ into $\mathcal{N}$ disjoint areas. Call them $\mathcal{A}_i$, with $i\in\{1,2,\ldots,\mathcal{N}\}$, and call $\varphi(i)$ the number of coverage circles $\mathcal{A}_i$ belongs to. If we call $\mathcal{P}[\mathcal{A}_i]$ the probability that the source node is located in $\mathcal{A}_i$, the global outage probability $\xi = \mathcal{P}[SNR \leq \Gamma]$ has the general expression:
\be
 \xi =  \sum_{i=1}^{\mathcal{N}}\mathcal{P}[\mathcal{A}_i]\int_0^R\!\!\ldots\!\int_0^R \!\mathcal{P}\left[SNR(\delta_1,\delta_2,\ldots,\delta_{\varphi(i)}) \leq\Gamma\right] f_{d_1,d_2,\ldots,d_{\varphi(i)}}^{(i)}(\delta_1,\delta_2,\ldots,\delta_{\varphi(i)}) d\delta_1 d\delta_2 \ldots d\delta_{\varphi(i)}
 \label{eqoutage}
\ee

Note that $\varphi(i)$ may also be equal to 0, if the chosen parameters and topology do not guarantee full coverage of the deployment area. For users located in these areas, we consider that the outage event has probability 1. In general, $\mathcal{N}$ depends on the selected topology, and grows quadratically with the number of circles. We point out that, if interference is also taken into account, the SNR should be replaced with the Signal to Interference-plus-Noise Ratio (SINR). In this case, however, also the interference term should be averaged, by integrating over the position of the interferer(s).
Although a simplified interference model may be used, we notice that this model should take into account the correlation between interference levels at all the Base Stations connected with the source of the useful signal. Such a model is beyond the scope of this paper.

Equation (\ref{eqoutage}) requires the knowledge of three distributions:
\begin{itemize}
 \item distribution of the source node position, which is uniform by assumption. This means that $\mathcal{P}(\mathcal{A}_i)$ is given by the ratio between the area of $\mathcal{A}_i$ and the area of the whole deployment area $\mathcal{A}_D$.
 \item distribution of the SNR, given the number of access points within transmission range and their distances from the source node. Once the distances are fixed, it follows from (\ref{distSNR}) that its distribution is the distribution of a finite sum of independent exponential random variables, whose parameters are related to the distances of the access points. This distribution is known.
 \item distribution $f_{d_1,d_2,\ldots,d_{\varphi(i)}}^{(i)}(\delta_1,\delta_2,\ldots,\delta_{\varphi(i)})$ of the distances between the source node and the access points within coverage range, dependent on the considered area $\mathcal{A}_i$. This distribution strongly depends on the shape of $\mathcal{A}_i$. Note that also for areas covered by only one access point, this distribution is no longer the distribution of the distance of a point randomly placed in a circular area from its center.
\end{itemize}

As to the topology, in this example we consider a cellular-like deployment of the access points, which are distributed in a hexagonal grid with side equal to $L$. Every other possible deployment is admissible; however, our choice leads to a lower computational burden due to the symmetries of the selected topology. In Figure \ref{fig:cerchiall}, an example of the considered topology is reported, for a fixed $L$ and $R$. The overall performance can be studied as a function of these two parameters, or equivalently, of $L$ and $P_M$. For the sake of power saving, lower values of $P_M$ are preferable, whereas, in order to reduce the network deployment cost, a higher $L$ is desirable. This in turn increases the outage probability, thus highlighting the need for a tradeoff.

\begin{figure}
    \centering
    \includegraphics[width=\figw]{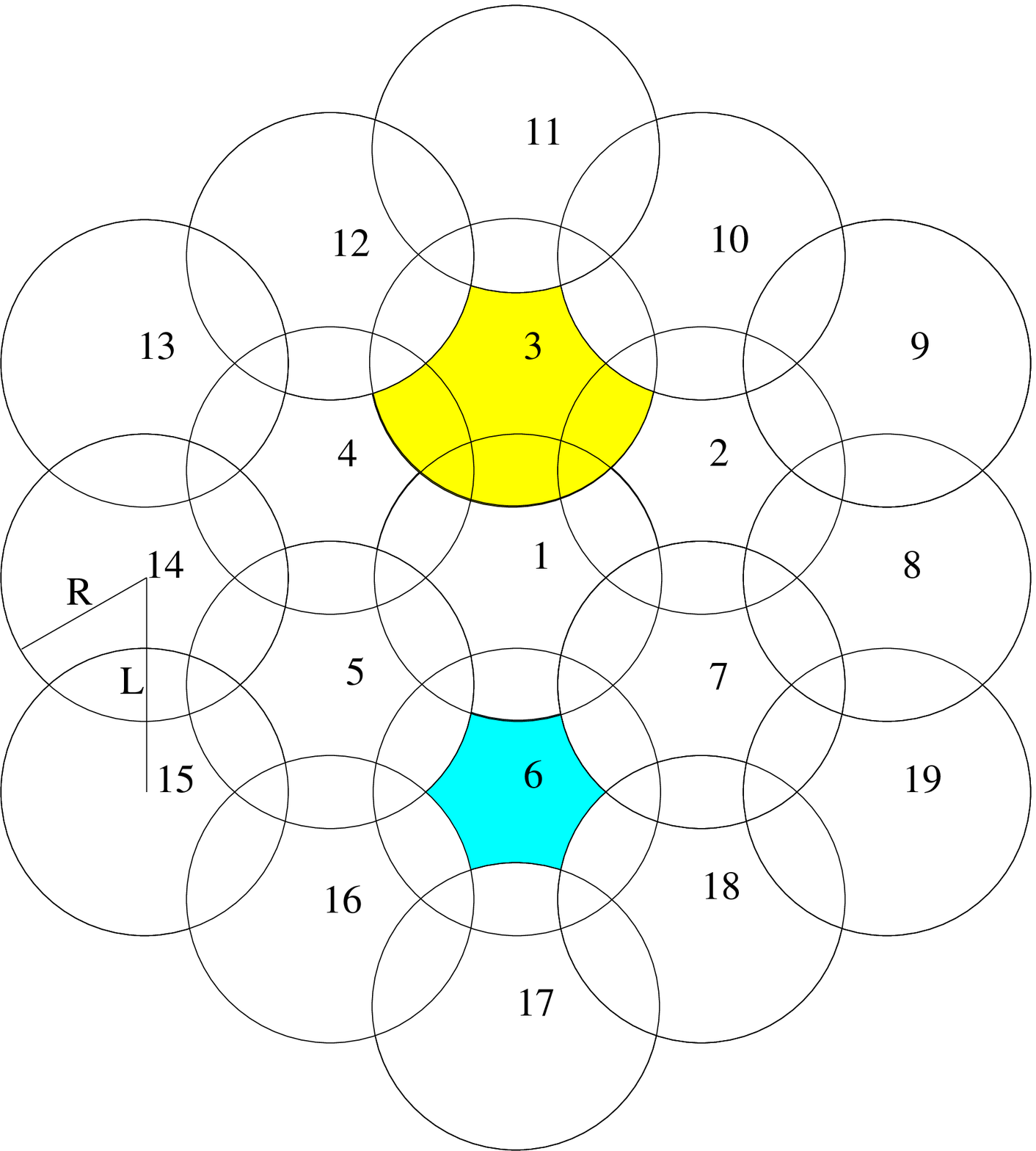}
     \caption{Example of plane tessellation for $\sqrt{3}/3L \leq R \leq \sqrt{3}/2L$. The shaded area around Access Point 3 represents the basic element of the tessellation. The shaded area around Access Point 6 is instead the portion of the area of each cell which is covered by only one Access Point.}
  \label{fig:cerchiall}
\vup\vup\vup\vup
\end{figure}

The proposed algorithm can be used to find $\mathcal{P}[\mathcal{A}_i]$ for each area $\mathcal{A}_i$. Moreover, it can also be used to find $f_{d_1,d_2,\ldots,d_{\varphi(i)}}^{(i)}(\delta_1,\delta_2,\ldots,\delta_{\varphi(i)})$, as we describe in the following.

Consider an area $\mathcal{A}_1$ covered by only one access point, for example the quasi hexagonal shaded area around access point 6 in Figure \ref{fig:cerchiall}. An analytical expression for the cumulative distribution function of the distance from the access point here is hardly derivable. However, we can approximate it by running the proposed algorithm several times, and properly selecting the radii of the circles. More specifically, we quantize the cdf with arbitarily small step $\rho$. To obtain the desired cdf, we keep the center of the circles and all the radii, except the radius of the circle around the access point 1, which is set to 0, and then increased by $\rho$ each time the algorithm is run. It is then sufficient to compute at each iteration the value of the area covered only by access point 1, and normalize it with the area of $\mathcal{A}_1$ to get the required cdf and, by numerical differentiation, the corresponding pdf. For areas covered by two or more circles, the same method is used, now 
properly varying the radii of the circles around the involved access points. Two and three-dimensional distributions are obtained for areas covered by two and three circles respectively. However, for areas covered by four or more circles, only three dimensional distributions are to be computed, since three distances uniquely determine the position of the source node, and hence all the other distances as well. Therefore, at most three integrals are to be numerically computed in (\ref{eqoutage}). Results become more and more precise for smaller values of the quantization step $\rho$, but the computational burden is also increased.

If we consider the cellular like distribution of Figure \ref{fig:cerchiall}, we can assume an infinite deployment area $\mathcal{A}_D$. In fact, due to the tessellation of the plane, the analysis can be limited to a finite region (the coloured region in Figure \ref{fig:cerchiall}), whose area can be regarded as the whole deployment area. With this choice, border effects are avoided, and a small number of circles are to be considered. The number of areas to be computed depends on the ratio between $L$ and $R$.

In Figure \ref{fig:probsucc_SNR10}, we report the success probability, that is, $\mathcal{P}[SNR \geq \Gamma]$, as a function of the transmission power $P_M$ and the distance $L$ between the access points, for $\Gamma = 10\spa\rm{dB}$. Once the target SNR $\Gamma$ is set, and the corresponding graph is plotted, it is possible to determine the required transmission power to achieve a given outage probability with a fixed access point density; vice versa, the minimum density to achieve the same outage probability can be found when the transmission power is instead fixed. Clearly, the success probability is higher when more power is available at the source, or when the density of the access points is increased. However, this comes at the cost of a higher number of required access points or of a shorter battery life of the source node. Depending on the relevance of these two costs, an objective function $g(.)$ could be also defined. Several choices are possible; as an example, the objective function may have the 
following form:
\be
 g(P_M, L) = \frac{1 - \xi(P_M, L)}{\mu P_M + \eta\mathcal{D}(L)}
 \label{valobie}
\ee
where the success probability appears at the numerator, and is derivable through (\ref{eqoutage}). At the denominator, $\mathcal{D}(L)$ is the access points density, equal to $2/(\sqrt{3}L^2)$, whereas $\mu$ and $\eta$ are normalizing constants, whose values can be properly selected. As an example, the value of $g$ as a function of $P_M$ and $L$ is depicted in Figure \ref{fig:probobie_SNR10} for specific values of $\mu$ and $\eta$. Although the success probability increases with $P_M$ and with $\mathcal{D}(L)$, the additional cost in terms of energy and number of deployed access points penalizes the choice of high $P_M$ and low $L$, thus identifying an optimal region.

\begin{figure}
    \centering
    \includegraphics[width=\figw]{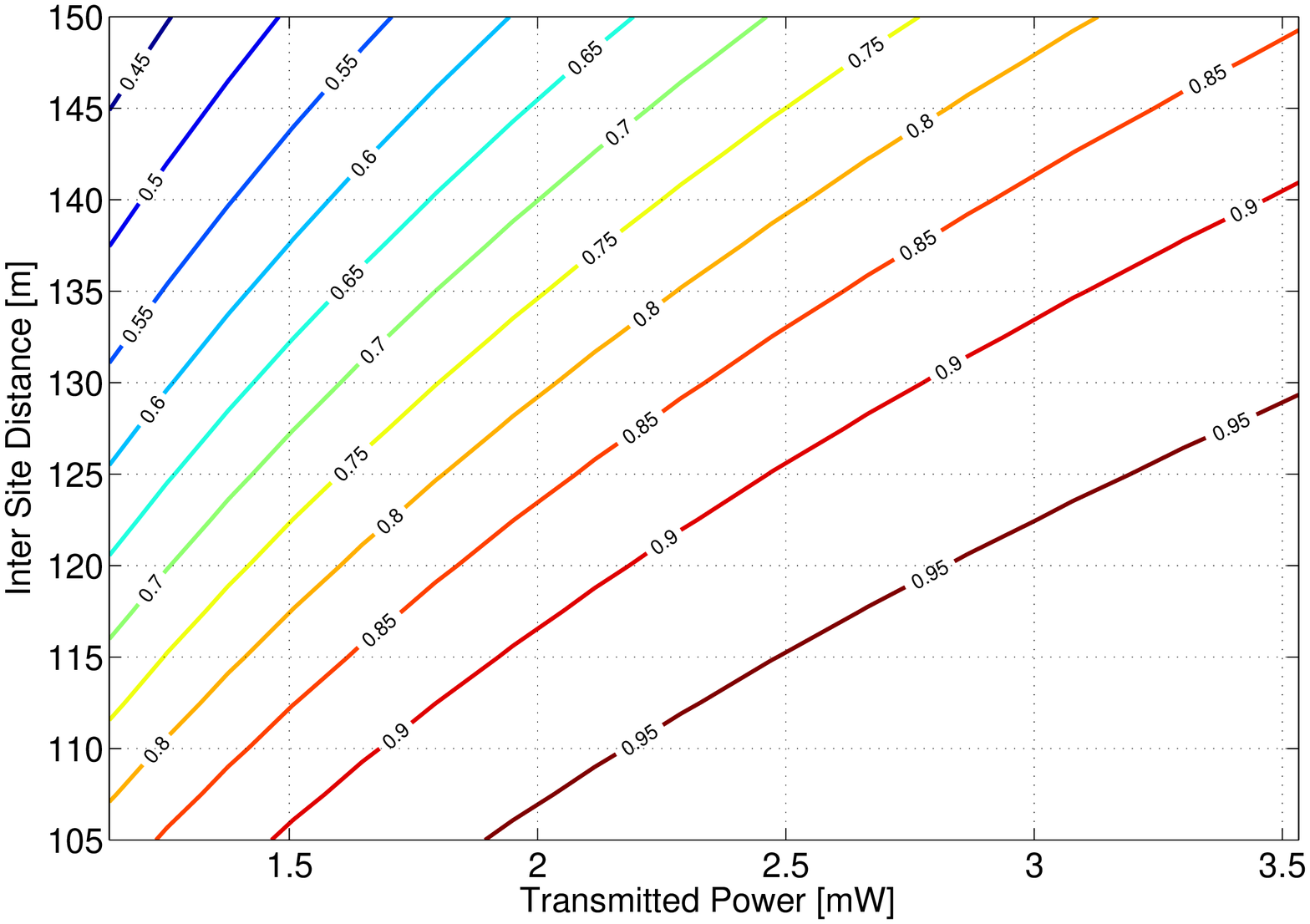}
     \caption{Success probability, as a function of $P_M$ and $L$. Here $\Gamma = 10\spa\rm{dB}$.}
  \label{fig:probsucc_SNR10}
\end{figure}

\begin{figure}
    \centering
    \includegraphics[width=\figw]{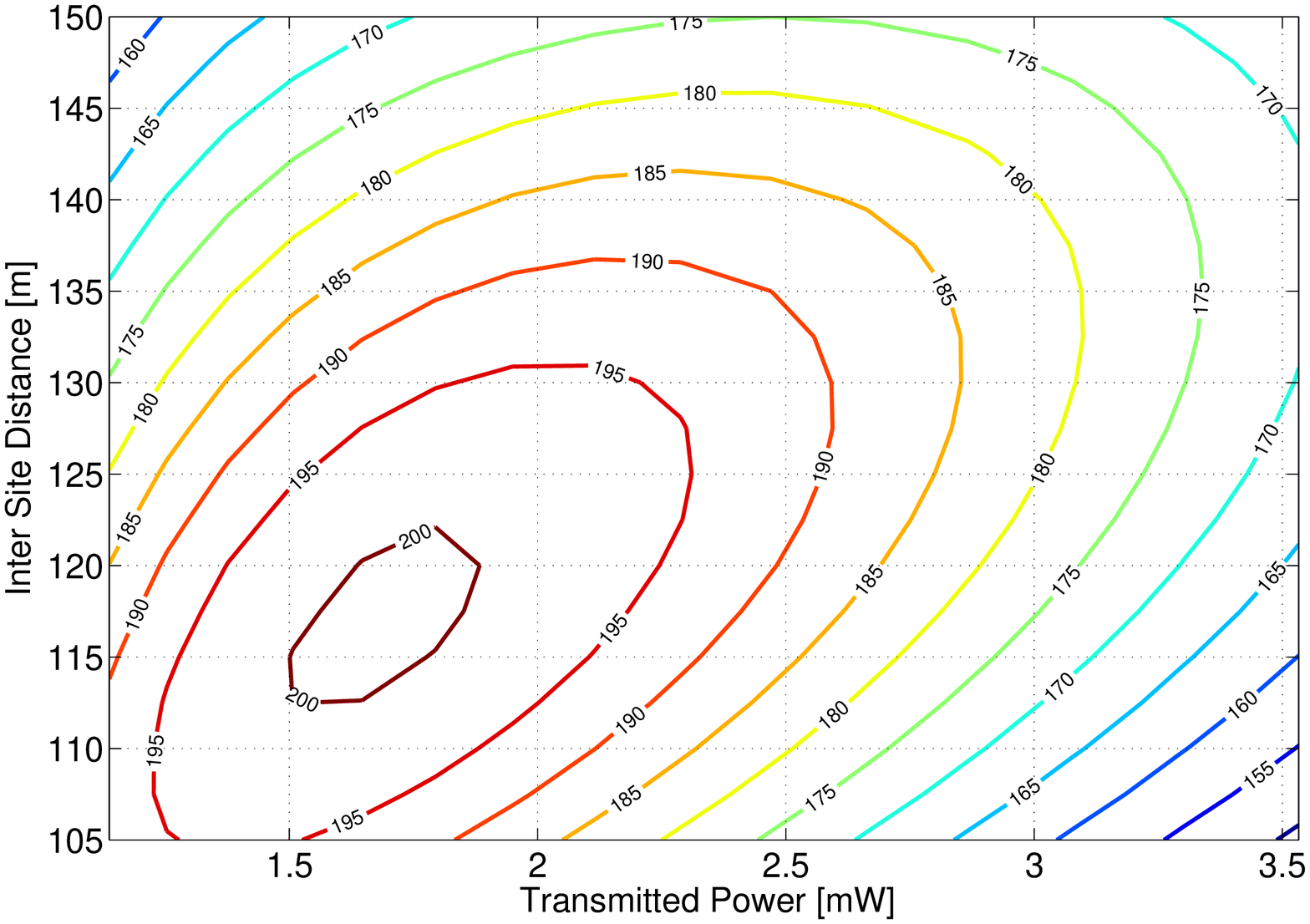}
     \caption{Values of the objective function $g$, as reported in (\ref{valobie}), as a function of $P_M$ and $L$. Here $\Gamma = 10\spa\rm{dB}$, $\mu = 1\spa\rm{W}^{-1}$, and $\eta = 30\spa\rm{m}^2$.}
  \label{fig:probobie_SNR10}
\end{figure}

\section{Conclusion}
In this paper, a practical algorithm for the computation of the intersection areas among any number of circles has been presented. The algorithm, based on two geometrical results, is designed to operate in an iterative manner, and takes advantage of a trellis structure to order and calculate all the required areas, given the radii and the mutual positions of the circles. An application of the algorithm has been presented in a network design problem, where cooperation is available among several access points. Our algorithm makes it possible to derive the distribution of the location of the source terminal, thus allowing the calculation of the outage probability as a function of the transmission power and the density of access points.

\appendices
\section{Proof of Theorem \ref{teo:intinun}}
\label{app:proof2}
\begin{proof}
Analogously to the proof of Theorem \ref{teo:existn}, we distinguish three cases:
\begin{enumerate}
 \item $\exists \{\gamma_{u_1},\ldots,\gamma_{u_w}\}{\subseteq} \mathcal{S}{:}\boldsymbol{\alpha}(u_1),\ldots,\boldsymbol{\alpha}(u_w) {=} \emptyset$;
 \item $\forall \gamma_i{\in} \mathcal{S} \; |\boldsymbol{\alpha}(i)| {=} 1$;
 \item $\forall \gamma_i{\in} \mathcal{S} \; \boldsymbol{\alpha}(i) {\neq} \emptyset$ and $\exists j:\; |\boldsymbol{\alpha}(j)| {>} 1$.
\end{enumerate}

We skip for now the first case, and focus on the following ones.

In the second case, each circle $\gamma_i$ has one arc belonging to $\Delta$. Since $m{\geq} 4$, there exist two circles $\gamma_t$, $\gamma_r{\in}\mathcal{S}$ such that $\alpha_t{\in}\boldsymbol{\alpha}(t)$ and $\alpha_r{\in}\boldsymbol{\alpha}(r)$ are non-consecutive sides of $\Delta$. Let us denote as
$P$ and $Q$ the points of intersection of $\gamma_t$ and $\gamma_r$. Note that the existence of $P$ and $Q$ is guaranteed by hypothesis, as $\mathcal{I}^{(n)}{\neq}\emptyset$. We also define $\alpha^r_t$ as the arc of $\Delta$ belonging to $\gamma_t$ and fully contained in $\gamma_r$,
and $\alpha^t_r$ as the arc of $\Delta$ belonging to $\gamma_r$ and fully contained in $\gamma_t$.

Moreover, there exists a circle $\gamma_h{\in}\mathcal{S}$ such that $P{\notin}\gamma_h$. In fact, if this circle did not exist,
then $P{\in}\mathcal{I}^{(n)}$ and the arcs $\alpha_r$ and $\alpha_t$ would be consecutive arcs of $\Delta$.
\begin{figure}
    \centering
    \includegraphics[width=\figwe]{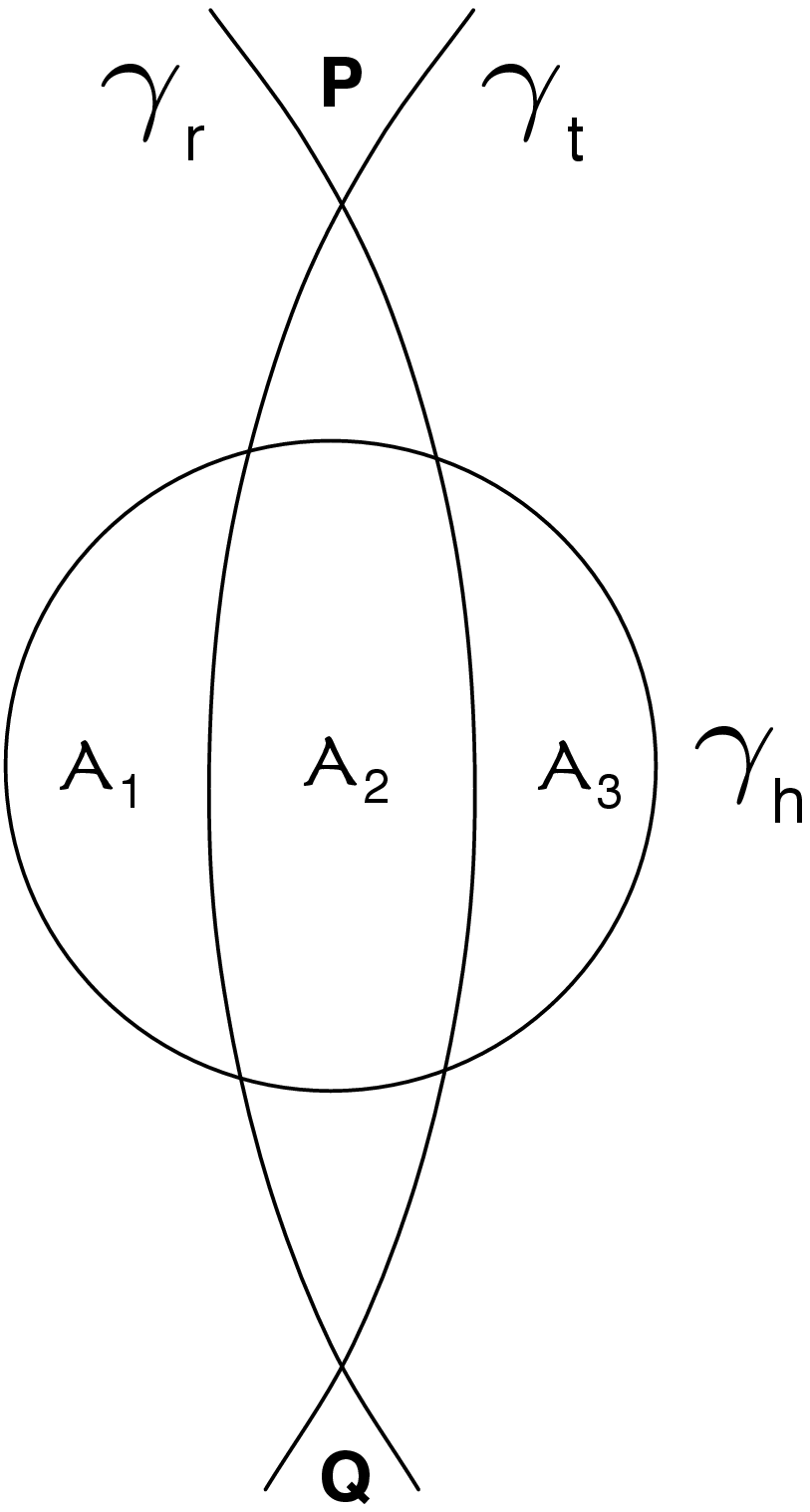}
     \caption{First configuration of the two possible when $m{=}n$, and each of the $n$ intersecting circles has exactly one arc of circumference belonging to the polygon $\Delta$. In the figure $\mathcal{A}_2$ contains $\Delta$.}
  \label{fig:sulteo21}
\end{figure}

We need to distinguish two configurations for point $Q$:
\begin{itemize}
\item In the first configuration, $Q$ does not belong to $\gamma_h$ (see Fig.~\ref{fig:sulteo21}). Therefore, the intersection of the circumferences of $\gamma_r$ and $\gamma_h$ belongs to the arc $\alpha^t_r$ and the intersection of the circumferences of $\gamma_t$ and $\gamma_h$
belongs to $\alpha^r_t$. As a consequence, $\gamma_h$ is divided into three disjoint areas $\mathcal{A}_1{\subset}\gamma_r$, $\mathcal{A}_2{\subset}\gamma_r{\cap}\gamma_t$ and $\mathcal{A}_3{\subset}\gamma_t$. We thus conclude that $\gamma_h{\subset}\gamma_r{\cup}\gamma_t$. Since, $\mathcal{I}^{(n-2)}(r,t)$ is contained in $\gamma_h$, then it is also included in $\gamma_r{\cup}\gamma_t$.
\item In the second configuration, $Q$ belongs to $\gamma_h$ (see Fig.~\ref{fig:sulteo22}). In this case there exists a $\gamma_x$ such that $Q{\notin}
\gamma_x$. If $P{\notin}\gamma_x$, the theorem can be proved as in the previous configuration. If $P{\in}\gamma_x$, then $\alpha_r^t$ and $\alpha_t^r$
belong to $\gamma_h{\cap}\gamma_x$, as they are sides of $\Delta$ and hence are part of $\mathcal{I}^{(n)}$. Thus, the points of intersection
between the circumferences of $\gamma_h$ and $\gamma_x$ cannot belong to $\gamma_r{\cap}\gamma_t$. It follows that $\gamma_r{\cap}\gamma_t$
is divided into three regions: $\mathcal{A}_1{\subset}\gamma_h$, $\mathcal{A}_2{\subset}\gamma_h{\cap}\gamma_x$ and $\mathcal{A}_3{\subset}\gamma_x$. Therefore, $\gamma_r{\cap}\gamma_t{\subset}\gamma_h{\cup}\gamma_x$, and since $\mathcal{I}^{(n-2)}(h,x){\subset}\gamma_r{\cap}\gamma_t$ the thesis is proved.
\end{itemize}
\begin{figure}
    \centering
    \includegraphics[width=\figw]{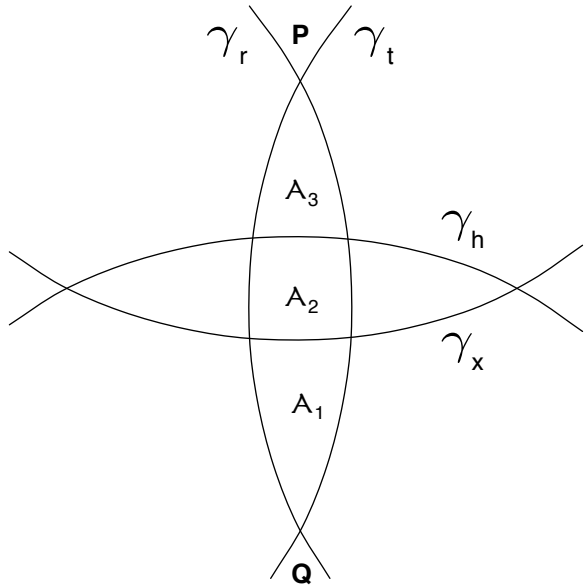}
     \caption{Second configuration of the two possible when $m{=}n$, and each of the $n$ intersecting circles has exactly one arc of circumference belonging to the polygon $\Delta$.}
  \label{fig:sulteo22}
\end{figure}

In the third case, consider the circle $\gamma_j$ with $|\boldsymbol{\alpha}(j)|{>}1$. In the proof of Theorem~\ref{teo:existn}, we showed that
the two arcs $\Delta$ adjacent to $\alpha_t{\in}	\boldsymbol{\alpha}(j)$ must belong to the circumference of two different circles, say $\gamma_h$
and $\gamma_x$. By construction, the arc $\alpha_t$ is contained in $\gamma_h{\cap}\gamma_x$. Moreover, also $\alpha_r{\in}\boldsymbol{\alpha}(j)$, with $\alpha_r{\neq}\alpha_t$, is contained in $\gamma_h{\cap}\gamma_x$, as $\alpha_r$ is an arc of $\Delta$. Therefore, $\gamma_j$ can be divided into three
disjoint regions analogous to those of the previous case, namely  $\mathcal{A}_1{\subset}\gamma_x$, $\mathcal{A}_2{\subset}\gamma_x{\cap}\gamma_h$ and $\mathcal{A}_3{\subset}\gamma_h$, which implies $\gamma_j{\subset}\gamma_x{\cup}\gamma_h$. Since $\mathcal{I}^{(n-2)}(x,h){\subset}\gamma_j$, the thesis in the third case is proved.

We now conclude the proof with the first case. Here, there exist $w$ circles containing $\mathcal{I}^{(n)}_{\{i_1,\ldots,i_n\}}$, with $1{\leq} w{\leq} n{-}3$. Define the subset $\mathcal{S}^{(n-w)}{=}\mathcal{S}{\setminus}\{\gamma_{u_1},\ldots,\gamma_{u_{n-w}}\}{\subset}\mathcal{S}$ that contains the circles $\gamma_j$ such that $\boldsymbol{\alpha}(j) {\neq} \emptyset$, and their
intersection $\mathcal{I}^{(n-w)}(u_1,\ldots,u_{n-w})$. Thus, there are $n{-}w$ circles with a non-empty set of arcs $\boldsymbol{\alpha}(j)$ and
$w$ circles that fully contain $\mathcal{I}^{(n)}_{\{i_1,\ldots,i_n\}}$. This case is equivalent to the second
or third case, if we consider only the circles with a non-empty set of arcs of $\Delta$. Since we have already shown that the theorem holds in those cases, we have here that $\exists \gamma_t$, $\gamma_r{\in}\mathcal{S}^{n-w}:\; \mathcal{I}^{(n-w-2)}(t,r) {\subset} \gamma_t{\cup}\gamma_r$. Therefore, 
\begin{equation}
\mathcal{I}^{(n-2)}(r,t) \subset \mathcal{I}^{(n-w-2)}(r,t) \subset \gamma_r{\cup}\gamma_t,
\end{equation}
In fact, $\mathcal{I}^{(n-2)}(r,t)$ is equal to the intersection between $\mathcal{I}^{(n-w-2)}(r,t)$  and $\bigcap_{i\in\mathcal{S}\setminus\mathcal{S}^{(n-w)}}\gamma_i$, and is thus a subset of $\mathcal{I}^{(n-w-2)}(r,t)$. That proves the theorem in this case.
Note that $w{\leq}n{-}3$, in fact in order to have
$m{\geq}4$ arcs, we need $\Delta$ to be bounded by at least three circles. 
\end{proof}

\section{Derivation of $\textbf{A}_{N_c}$ for $N_c = 4$}
\label{app:specialcase}
We describe here how $\textbf{A}_{N_c}$ can be computed when $N_c=4$ and when an additional geometric check is necessary.
The available vectors in this case are $\hat{\textbf{A}}_1$, $\hat{\textbf{A}}_2$ and $\hat{\textbf{A}}_3$, whose expressions have been reported in (\ref{varia}). It is useful to consider $\hat{\textbf{A}}_2^* = - \hat{\textbf{A}}_2$. In this manner, all the elements of $\hat{\textbf{A}}_2^*$ are smaller than or equal to $\mu$, whereas all the elements of $\hat{\textbf{A}}_1$ and $\hat{\textbf{A}}_3$ are greater than or equal to $\mu$. Therefore, we define
\be
 a_{\gamma} = \min\left(\hat{\textbf{A}}_1\right),\:
 b_{\gamma} = \max\left(\hat{\textbf{A}}_2^*\right),\:
 c_{\gamma} = \min\left(\hat{\textbf{A}}_3\right)
\ee
According to the considerations reported above, if $m\geq 4$ at least one of the exclusive intersection areas between 2 circles is 0, and $\mu$ is simply equal to $b_{\gamma}$. The same holds also if $m = 1$ and $m = 2$, so the only case that must be studied is when $m$ is equal to 3.
\begin{figure}
    \centering
    \includegraphics[width=\figw]{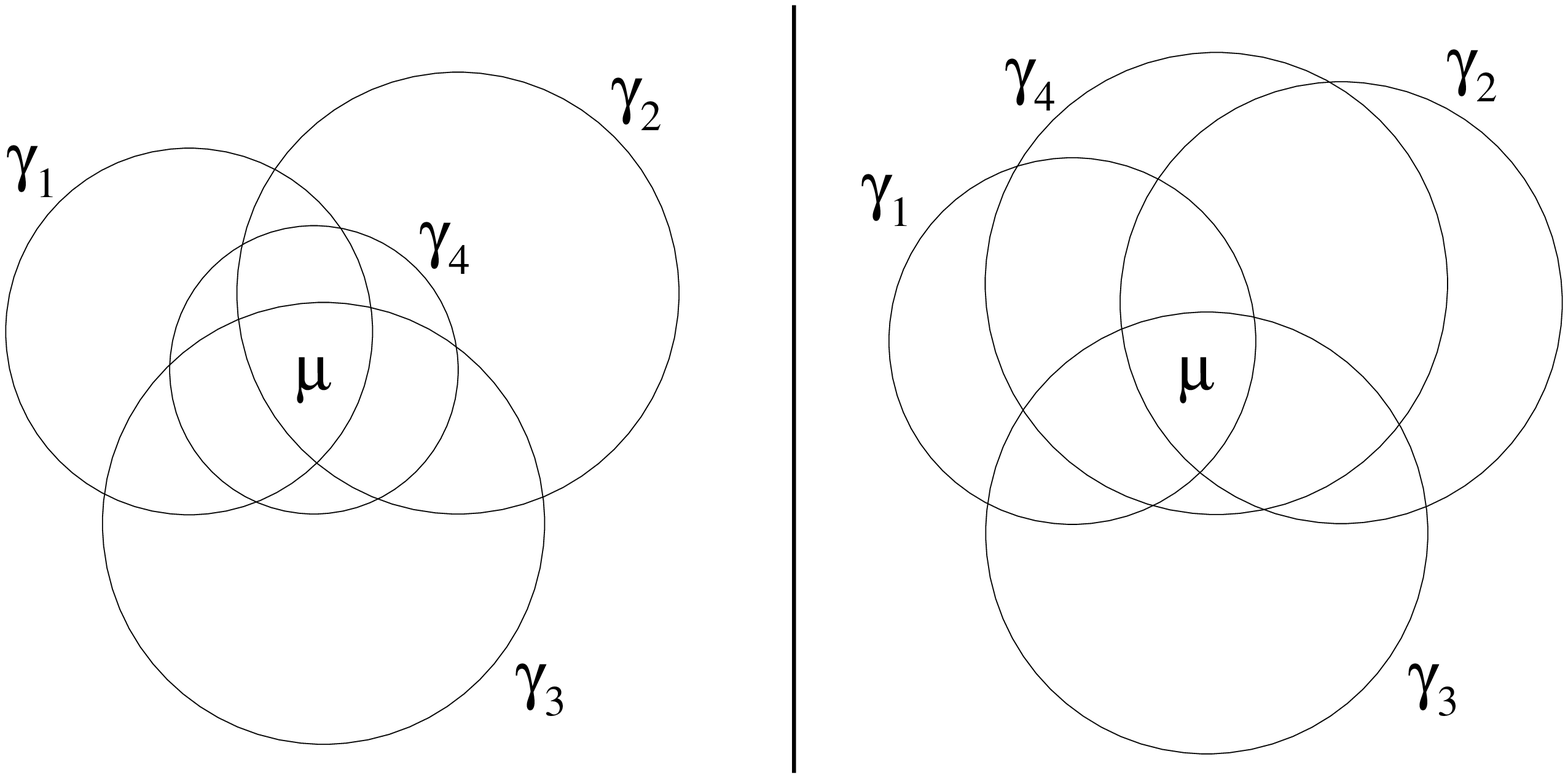}
     \caption{Possible deployments of 4 circles when $m$ is equal to 3. On the left, $\gamma_4$ is included in the union of the other three circles, whereas on the right this is not true.}
  \label{fig:casoraro}
\end{figure}
Assume, without loss of generality, that the circle $\gamma_4$ contains the intersection of $\gamma_1$, $\gamma_2$ and $\gamma_3$. Then, there are two possible cases, as reported also in Figure \ref{fig:casoraro}:
\begin{itemize}
 \item $\gamma_4$ is contained in the union of the other three circles. In this case, $\mu^*_4 = 0$, and hence $\mu = a_{\gamma}$;
 \item $\gamma_4$ is not contained in the union of the other three circles. In this case, it can be shown that it fully contains the exclusive intersection of two circles, meaning that $\mu = b_{\gamma}$.
\end{itemize}

Moreover, since the intersection area of $\gamma_1$, $\gamma_2$ and $\gamma_3$ is included in $\gamma_4$, also $\mu^*_{1,2,3} = 0$, and in both cases also $\mu = c_{\gamma}$. In this manner, we have proved that in any case the unknown value of $\mu$ is equal to one of the three values $a_{\gamma}$, $b_{\gamma}$ or $c_{\gamma}$, as summarized also in Table \ref{tabvalues}.
\begin{table}
 \centering
 \caption{Relationships among $a_{\gamma}$, $b_{\gamma}$ and $c_{\gamma}$, depending on the positions of the circles.}
\begin{tabular}{l|c|c}
\hline
Case & Ordering & Value of $\mu$ \\
\hline
$m \neq 3$ & $a_{\gamma},c_{\gamma}\geq b_{\gamma}$ & $b_{\gamma}$ \\
$m = 3$ and $\gamma_4 \subset \bigcup_{i=1}^3\gamma_i$ & $a_{\gamma}=c_{\gamma}\geq b_{\gamma}$ & $a_{\gamma}(=c_{\gamma})$ \\
$m = 3$ and $\gamma_4 \not\subset \bigcup_{i=1}^3\gamma_i$ & $a_{\gamma} \geq b_{\gamma} = c_{\gamma}$ & $b_{\gamma}(=c_{\gamma})$ \\
\hline
\end{tabular}
\label{tabvalues}
\end{table}
However, it is still to be determined how the algorithm can recognize which one of the three terms is the actual value. In most cases, this can be inferred by the relationships among their values. We can distinguish the following exhaustive possibilities:
\begin{itemize}
 \item $a_{\gamma} \neq b_{\gamma} \neq c_{\gamma}$: in this case, looking at Table \ref{tabvalues}, it follows that $m \neq 3$, and therefore $\mu = b_{\gamma}$;
 \item $a_{\gamma} = b_{\gamma} \neq c_{\gamma}$ or $a_{\gamma} \neq b_{\gamma} = c_{\gamma}$: in both these cases, recalling that $a_{\gamma},c_{\gamma} \geq \mu$ whereas $b_{\gamma}\leq\mu$, it follows that necessarily $\mu = b_{\gamma}$;
 \item $a_{\gamma} = b_{\gamma} = c_{\gamma}$: as in the previous case;
 \item $a_{\gamma} = c_{\gamma} \neq b_{\gamma}$: this is the only case where it is not possible to determine whether $\mu = b_{\gamma}$ or $\mu = a_{\gamma}$. In fact, looking at Table \ref{tabvalues}, this may happen both if $m \neq 3$ and if $m = 3$ and $\gamma_4 \subset \bigcup_{i=1}^3\gamma_i$: in the former case, $\mu = b_{\gamma}$, in the latter instead $\mu = a_{\gamma}$.
\end{itemize}

It is clear that, in the last case, it is very unlikely that $m \neq 3$, since this would mean that the smallest exclusive intersection among 1 circle has the same (nonzero) area as the smallest exclusive intersection of three circles (since $a_{\gamma} = c_{\gamma}$). Anyway, since this may happen, it is necessary to determine the value of $m$ in a different way.

\begin{figure}
    \centering
    \includegraphics[width=\figwe]{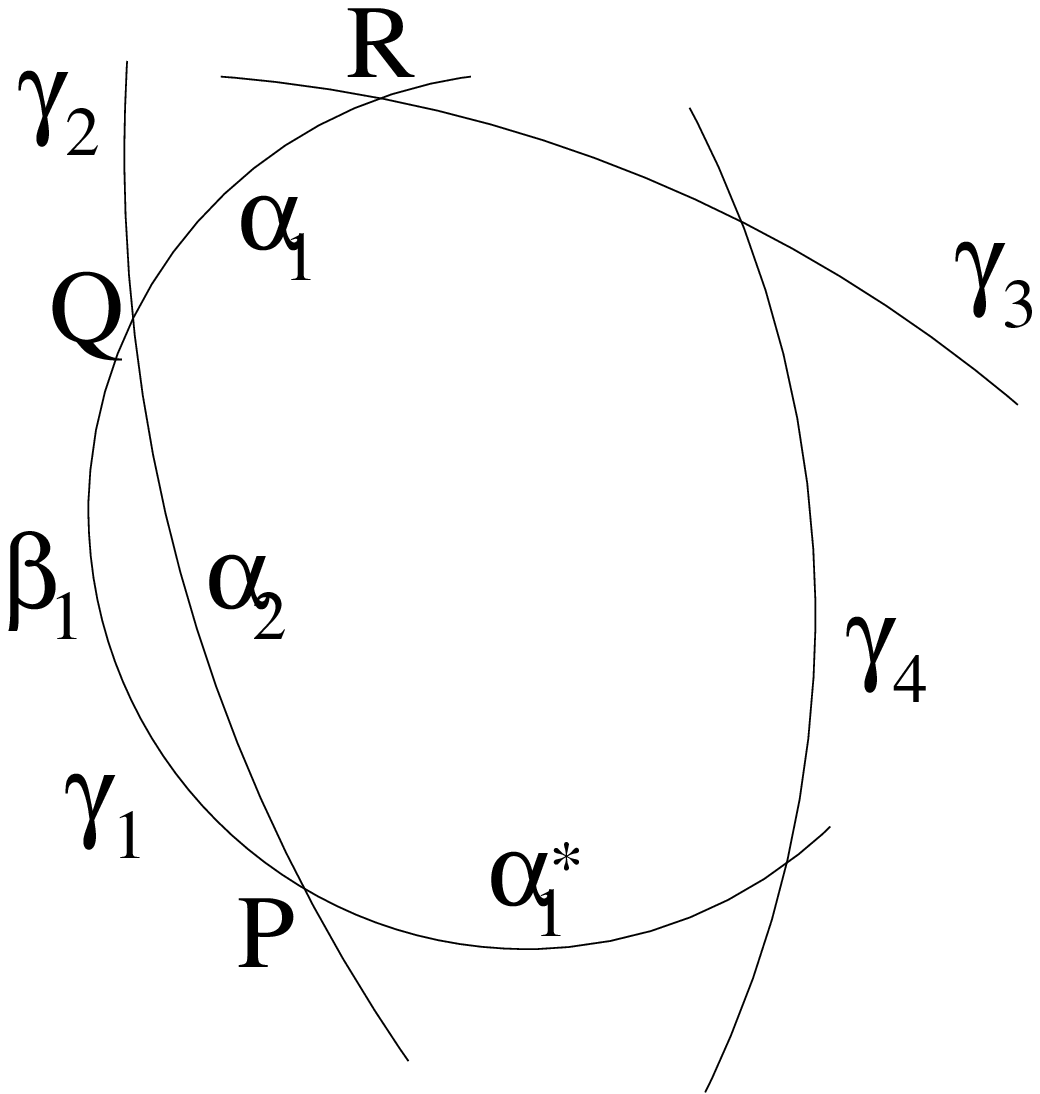}
     \caption{Intersection of 4 circles, with $m = 5$.}
  \label{fig:interm5}
\end{figure}

We first of all calculate which values of $m$ may effectively result in $a_{\gamma} = c_{\gamma} \neq b_{\gamma}$. Since it can be shown that $1\leq m\leq 2(N_c-1)$, for $N_c = 4$ we have $1\leq m\leq 6$. We exclude the following cases:
\begin{itemize}
 \item $m = 1$: in this case $\exists i$: $\gamma_i \subset \gamma_k$, $\forall k \in \{ 1,2,3,4 \}$, $k\neq i$. This implies that $\mu^*_i = 0$. Moreover, $\forall j\neq i$, we have $\gamma_i \cap\gamma_j \subset \gamma_k$, $\forall k \in \{1,2,3,4\}$, $k\neq i,j$. Hence, also $\mu^*_{i,j} = 0$, which implies that $a_{\gamma} = b_{\gamma}$.
 \item $m = 2$: in this case $\exists i,j$, $i\neq j$: $\gamma_i \cap \gamma_j \subset \gamma_k$, $\forall k \in \{1,2,3,4\}$, $k \neq i,j$. This means that $\mu^*_{i,j} = 0$. In addition, $\forall k \neq i,j$, it is also true that $\gamma_i \cap \gamma_j \cap \gamma_k \subset \gamma_p$, $\forall p \in \{1,2,3,4\}$, $p\neq i,j,k$, which in turn implies that also $\mu^*_{i,j,k} = 0$. As a consequence, $b_{\gamma} = c_{\gamma}$.
 \item $m = 4$, with 2 sides belonging to the same circle: in this case Theorem \ref{teo:intinun} holds, meaning that there exists one exclusive intersection among two circles with zero area. Moreover, it is clear that there exists a circle $\gamma_i$ that fully contains the intersection of the other three circles (the ones whose arcs delimit the intersection area of all the four circles), meaning that $\mu^*_{j,k,p} = 0$, with $j,k,p \in \{1,2,3,4\}$, and $i \neq j \neq k \neq p$. Therefore $b_{\gamma} = c_{\gamma}$.
 \item $m = 5$: here again Theorem \ref{teo:intinun} holds, meaning that $\exists i,j \in \{1,2,3,4\}$, $i \neq j$: $\mu^*_{i,j} = 0$. In addition, as reported in Figure \ref{fig:interm5}, it is clear that two non consecutive arcs $\alpha_1$ and $\alpha_1^*$ delimiting the intersection area of the four circles belong to the same circle, say $\gamma_1$. Assume that the arc between them, namely $\alpha_2$, belongs to $\gamma_2$. This arc divides the circle $\gamma_1$ in two parts. The one containing the intersection of the four circles is fully contained in $\gamma_2$, since the two circumferences $\hat{\gamma}_1$ and $\hat{\gamma}_2$ cannot intersect in more than two points, and they already intersect in $P$ and $Q$. The other part of $\gamma_1$ is contained in $\gamma_3$ (as well as in $\gamma_4$). In fact, if this were not true, since $P$ and $Q$ both belong to $\gamma_3$, the arc $\beta_1$ should intersect the circumference $\hat{\gamma}_3$ in two points. Since $\hat{\gamma}_1$ and $\hat{\gamma}_3$ already 
intersect in $R$, this would cause them to intersect in more than two points,which is not possible. Therefore, $\gamma_1 \subset \gamma_2 \cup \gamma_3$, and $\mu^*_1 = 0$, meaning that $a_{\gamma} = b_{\gamma}$.
 \item $m = 6$: the same reasoning as for $m = 5$ can be done in this case, again resulting in $a_{\gamma} = b_{\gamma}$.
\end{itemize}

Having excluded all the above listed cases, there are only two possible deployments that may result in $a_{\gamma} = c_{\gamma} \neq b_{\gamma}$: either $m = 3$ or $m = 4$ with all the four arcs belonging to different circles. They are reported in Figure \ref{fig:last2cases}.
\begin{figure}
    \centering
    \includegraphics[width=\figwd]{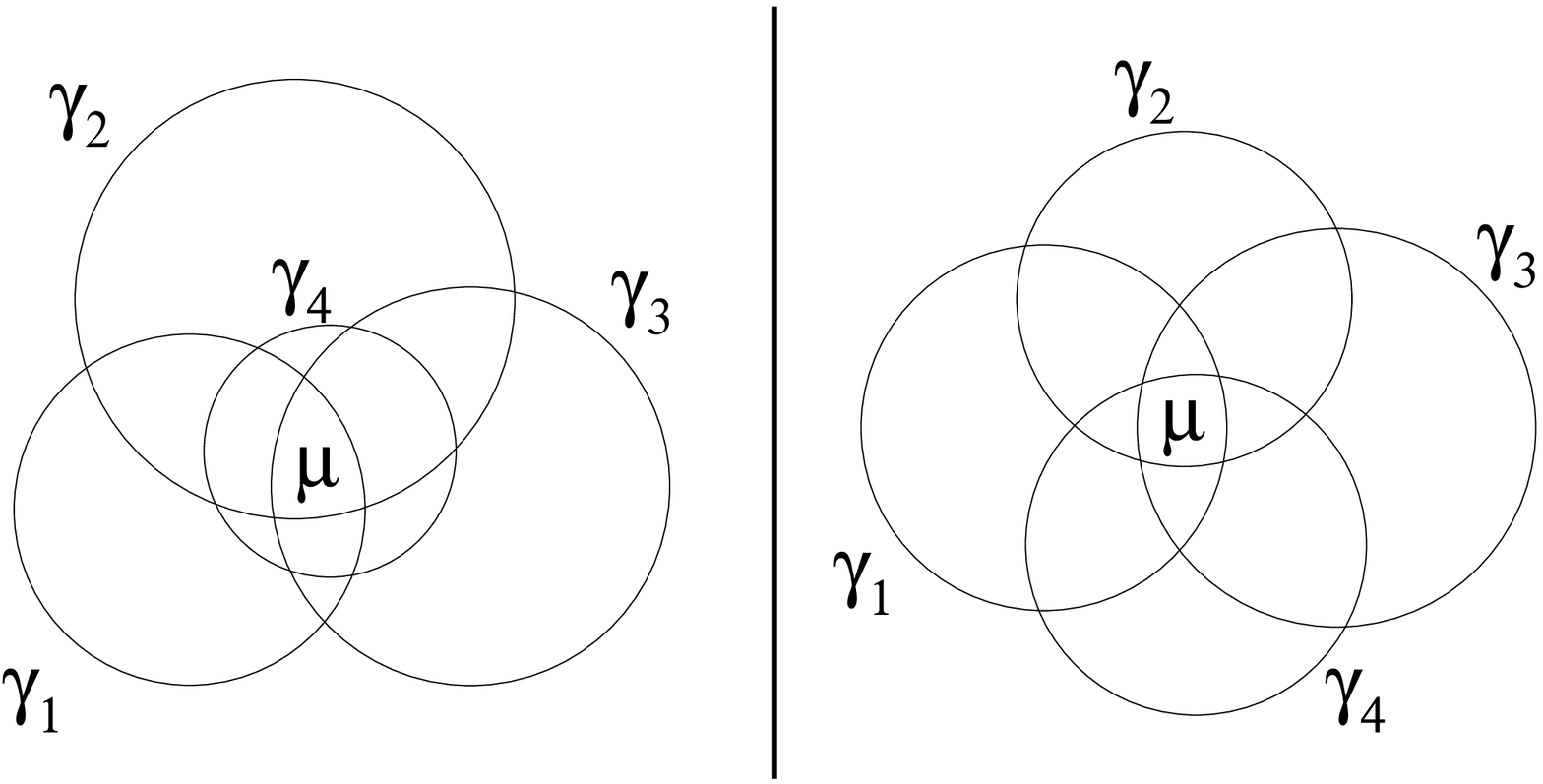}
     \caption{The only two possible deployments of 4 circles such that $a_{\gamma} = c_{\gamma} \neq b_{\gamma}$.}
  \label{fig:last2cases}
\end{figure}

Recall that if $m = 3$, the only possible case is the one reported in Figure \ref{fig:casoraro} on the left (and depicted also in Figure \ref{fig:last2cases} on the left). The deployment depicted in Figure \ref{fig:casoraro} on the right is instead not possible, since it implies $b_{\gamma} = c_{\gamma}$, as reported also in Table \ref{tabvalues}.

From the previous investigation, we can conclude that when $a_{\gamma} = c_{\gamma} \neq b_{\gamma}$, the value of $\mu$ is $a_{\gamma}$ if the circles are deployed as in Figure \ref{fig:last2cases} on the left, and is equal to $b_{\gamma}$ if the circles are deployed as in Figure \ref{fig:last2cases} on the right. No other deployments are compatible with the given inequality among $a_{\gamma}$, $b_{\gamma}$ and $c_{\gamma}$. The straightforward way to distinguish between the two cases is to calculate all the twelve intersection points between the four circles. In both situations each circle contains exactly three points of intersection between the other three circles. More precisely: $\forall i \in \{1,2,3,4\}$, $\exists P,Q,R$: $P \in \hat{\gamma}_j\cap\hat{\gamma}_k$, $Q \in \hat{\gamma}_j\cap\hat{\gamma}_p$ and $R \in \hat{\gamma}_k\cap\hat{\gamma}_p$, with $j,k,p \in \{1,2,3,4\}$, and $i$, $j$, $k$ and $p$ all different from each other, such that $P,Q,R \in \gamma_i$, where again $\hat{\gamma}_i$ 
indicates the circumference of circle $\gamma_i$.

However, if $m = 3$ there is one circle that contains the intersection area of the other three circles, meaning that there is one (and only one) circle $\gamma_i$ such that, using the same notation as above, $P \in \gamma_p$, $Q \in \gamma_k$ and $R \in \gamma_j$. This is not true when $m = 4$, since in that case, one of the three points contained in each circle $\gamma_i$, given by the intersection of two circumferences $\hat{\gamma}_j$ and $\hat{\gamma}_k$, does not belong to the third circle $\gamma_p$. Simple geometric comparisons among the distances between centers and intersection points and the radii of the circles are then enough to distinguish the two cases and, finally, determine the correct value of $\mu$. Note that this check is necessary only for $N_c = 4$, and only when $a_{\gamma} = c_{\gamma} \neq b_{\gamma}$.

\bibliographystyle{IEEEtran}
\bibliography{IEEEabrv,wibib.bib}
\end{document}